\documentclass[journal,twoside,web]{ieeecolor}
\usepackage{generic}
\usepackage{cite}
\usepackage{amsmath,amssymb,amsfonts}
\usepackage{algorithmic}
\usepackage{graphicx}
\usepackage{algorithm,algorithmic}
\usepackage{hyperref}
\usepackage{textcomp}
\usepackage{graphicx}

\newtheorem{definition}{Definition}
\newtheorem{theorem}{Theorem}

\newtheorem{lemma}{Lemma}

\def\BibTeX{{\rm B\kern-.05em{\sc i\kern-.025em b}\kern-.08em
    T\kern-.1667em\lower.7ex\hbox{E}\kern-.125emX}}
\allowdisplaybreaks

\begin{document}
\title{Gain Scheduling with a Neural Operator for\\ a Transport PDE with Nonlinear Recirculation}
\author{Maxence Lamarque, Luke Bhan, Rafael Vazquez, and Miroslav Krstic
\thanks{Maxence Lamarque is with \'{E}cole des Mines
de Paris, France, {maxence.lamarque@etu.minesparis.psl.eu  }}
\thanks{Luke, Bhan and Miroslav Krstic are with the University of California, San Diego, {lbhan@ucsd.edu, krstic@ucsd.edu}} 
\thanks{Rafael Vazquez is with the Department of Aerospace Engineering Universidad de Sevilla, {rvazquez1@us.es}}
\thanks{The second author is supported by DOE grant DE-SC0024386 and the work of the
fourth author was funded by FOSR grant FA9550-23-1-0535 and NSF
grant ECCS-2151525.}
}

\maketitle

\begin{abstract}
To  stabilize PDE models, control laws require space-dependent functional gains mapped by nonlinear operators from the PDE functional coefficients. 
When a PDE is nonlinear and its ``pseudo-coefficient'' functions are state-dependent, a gain-scheduling (GS) nonlinear design is the simplest approach to the design of nonlinear feedback. The GS version of PDE backstepping employs gains obtained by solving a PDE at each value of the state. Performing such PDE computations in real time may be prohibitive. The recently introduced neural operators (NO) can be trained to produce the gain functions, rapidly in real time, for each state value, without requiring a PDE solution. In this paper we introduce NOs for GS-PDE backstepping. GS controllers act on the premise that the state change is slow and, as a result,  guarantee only local stability, even for ODEs. We establish local stabilization of hyperbolic PDEs with nonlinear recirculation using both a ``full-kernel'' approach and the ``gain-only'' approach to gain operator approximation. Numerical simulations illustrate stabilization and demonstrate speedup by three orders of magnitude over traditional PDE gain-scheduling. Code (\href{https://github.com/lukebhan/NeuralOperatorApproximatedGainScheduling}{Github}) for the numerical implementation is published to enable  exploration. 

{\normalfont\sf\color{blue}This work was the subject of the 2023 Bode Prize Lecture by the last coauthor.} 
\end{abstract}

\section{Introduction}
\paragraph{Goals and results} We present a model-based, gain-scheduled (GS) backstepping design for the control of a hyperbolic PDE with nonlinear recirculation, where the gain kernels are computed via a Neural Operator (NO). We leverage the recent breakthrough on deep neural network-based operator approximations (DeepONet \cite{Lu2021DeepONet}) to prove that our model-based GS design is stabilizable and real-time implementable due to an $\sim 10^3\times$  speedup thanks to the NO. 

This is the first  result in which the NO is used to recompute the kernel \emph{online, at every time step}, based on learning a neural operator  of the kernel  only \emph{once, offline}. 

Even an exact gain scheduler is only locally stabilizing, and even local stabilization is complex to prove for a nonlinear PDE. We prove that both our exact and our NO-approximated GS-PDE backstepping controllers are stabilizing. 


\paragraph{Gain scheduling}
In the GS literature for nonlinear ODEs we single out the classic papers \cite{103361,58498} and the survey \cite{RUGH20001401}, which makes links with linear parameter varying systems and highlights the need to limit the state derivative, with only local stabilization being generally possible with GS.

\paragraph{Stabilization of nonlinear PDEs using backstepping-based gain-scheduling}
For the linear version of the hyperbolic PDEs we study here was introduced in \cite{krstic2008Backstepping}.  
The first and only attempt at global stabilization of a class of {\em general} nonlinear PDE systems 
was made in \cite{VAZQUEZ20082778,VAZQUEZ20082791}, using infinite Volterra series. The complexity of that result highlights the importance of seeking methodologies that fall in between such  complex  nonlinear designs and  linear designs such as \cite{Coron2013Local}. In fact, for parabolic PDEs with nonlinear growth exceeding $|u|\log^2(1+|u|)$, \cite{AIHPC_2000__17_5_583_0} shows that finite-time blow up cannot be prevented with boundary control, hence, the pursuit of global stabilization is in vain, in general. 
A few CLF-based control designs for nonlinear PDEs include \cite{9744516,doi:10.1137/19M1252235}. Only a few particular cases admit a linearizing transformation \cite{krstic2008nonlinear,bekiaris2019nonlinear}. Hardly any results exist on GS for PDEs. The paper \cite{203622} focuses on finite-dimensional approximations. In \cite{doi:10.1137/20M1341167}, GS is used for time-varying (not nonlinear) PDEs. Our paper builds on \cite{SiranosianAntranikA2011GSBC},  where  exact---and explicit---GS for a special example of hyperbolic PDEs with recirculation is considered. 

\paragraph{Learning approximations for model-based PDE control with stability guarantees}
Recent breakthroughs in the mathematical properties of neural networks \cite{392253}, \cite{Lu2021DeepONet}, \cite{lanthaler2023nonlocal} have spawned an exploration of functional operator approximations in control theory. This was first explored for control of transport PDEs in \cite{bhan_neural_2023} and extended to reaction-diffusion PDEs and observers in \cite{krstic2023neural}. PDE backstepping  \cite{krstic2008boundary}, \cite{krstic2008Backstepping} kernels---typically governed by PDEs---were replaced with neural operator approximations. Under the universal approximation theorem \cite{lu2021advectionDeepONet}, the papers \cite{bhan_neural_2023,krstic2023neural} prove that the  PDE plant is stabilized provided the operator approximates the true kernel with required accuracy. Recently, \cite{qi2023neural} and \cite{wang2023deep} extended these initial results to hyperbolic and parabolic PDEs with delays, where the  kernels are governed by multiple PDEs. 

\paragraph{Paper  outline} In Sec. \ref{sec:exactgs}, we introduce the exact GS controller for the transport PDE with recirculation. 
 In Sec. \ref{sec:existenceofk}, we establish the existence and continuours differentiability of the gain-schedulable backstepping kernel to be approximated by the DeepONet. In Sec. \ref{sec:approximatek}, we define the kernel operator and present its DeepONet approximation. In Section \ref{sec:deeponetstabilizes}, we present the main result: GS with DeepONet-approximated kernel is locally stabilizing. In Sec. \ref{sec:H1target} and \ref{sec:normequiv}, we prove local stabilization in $H_1$ for the perturbed target system under the DeepONet approximated kernel and then convert the perturbed system back into the original plant variables. In Sec. \ref{sec:gain-only}, we present an alternative analysis that only employs only the evaluation of the gain kernel at the boundary
 and show the result is still locally stabilizing. Lastly, Sec. \ref{sec:simulations} presents simulations illustrating the theoretical stabilization results. 

\paragraph{Contributions}
While \cite{bhan_neural_2023} only gave a glimpse at 
possibilities of
NOs, only linear PDEs have been considered so far. This paper takes up the NOs and GS simultaneously for the first time. Even our {\em exact GS} result, without the NO (obtained by setting $\epsilon=0$ in our results and not belabored as a separate statement), is the first GS result for a nonlinear PDE, not counting the simple explicit GS example in \cite{SiranosianAntranikA2011GSBC}. When we combine GS with  NO, two perturbations arise simultaneously in our analysis: (1) the perturbation from treating the state as quasi-constant, which even without an NO limits the achievable stabilization to local, and (2) the perturbation from the approximation of the GS kernel using an NO. Seeing the complexity of the analysis, the reader will understand why we pursue in this initial GS paper only a system where the ``scheduling variable'' is scalar (the  state value at the outlet). Jumping straight into a general PDE class \cite{VAZQUEZ20082778,VAZQUEZ20082791}, 
with a ``scheduling variable'' that is a {\em function} of space, would  obscure the concept of simultaneously dealing with the PDE nonlinearities, GS design, and NO approximation. It would be incomprehensible and unpedagogical. 

Our  results include a DeepONet approximation theorem for a class of parametrized Volterra integral equations, $H_1$ local Lyapunov stability analyses for two approaches to approximation and backstepping transformation (the ``full-kernel'' approach, which employs the NO GS kernel in the backstepping transform, and the ``gain-only'' approach, which employs the exact GS kernel in backstepping), and simulations
with a computational speedup of $\sim 10^3\times$ and stabilization. 

\paragraph{Notation}
    We denote the convolution as
$     a * b (x, \nu) = \int_{0}^{x} a(x-y, \nu) b(y, \nu) dy$,    and use the notation $a_{x \nu}
= \frac{\partial^2 a}{\partial x \partial \nu} 
$ and
$    \|a(t)\| = \sqrt{\int_{0}^{1} a^2(x, t)dx}$ for a function $a$  on $[0, 1] \times \mathbb{R}^+$.

\section{Exact Gain-Scheduled (GS) PDE Backstepping for a Hyperbolic PDE with Nonlinear Recirculation} \label{sec:exactgs}

We study nonlinear hyperbolic PDE systems of the form
  \begin{eqnarray}
    u_t(x, t) &=& u_x(x, t) + \beta(x, u(0, t)) u(0, t) \,, \label{eq:ut_def} \\
    u(1, t) &=& U(t), \qquad (x, t) \in [0, 1] \times \mathbb{R}^+ \,, \label{eq:u1_def}
  \end{eqnarray}
  with the function $\beta \in \mathcal{C}^1([0,1]\times \mathbb{R}, \mathbb{R})$ denoting the nonlinear, $u(0,t)$-dependent, recirculation gain.
Our objective is to (locally) stabilize this system, at the equilibrium $u \equiv 0$, with a GS feedback  implemented through the boundary input $U$.  We ``pretend,'' at each time $t$, that the state trace $u(0,t)$ is constant and employ a gain  corresponding to such a ``constant $u(0,t)$.''

To obtain such a gain , wemap the system \eqref{eq:ut_def}, \eqref{eq:u1_def} into a transport PDE $w_t = w_x$ using the backstepping transformation 
  \begin{eqnarray}
    w(x, t) &=& u(x, t) - \int_{0}^{x}k(x-y, u(0, t)) u(y, t) dy\,, )
    \label{eq:w_true_k_1}
  \end{eqnarray}
 with a homogeneous inlet boundary condition $w(1,t)=0$, using the boundary control law
  \begin{equation}
   U(t) = \int_{0}^{1}k(1-y, u(0, t)) u(y, t)dy\,.
    \label{eq:U_def}
  \end{equation}
If $u(0,t)$ were indeed a constant $\nu \equiv u(0,t)$, the kernel would need to satisfy the  Volterra integral equation
    \begin{align}
     k(x, \nu) = -\beta(x, \nu) + \int_{0}^{x} \beta(x-y, \nu)k(y, \nu)dy\,, &
     \nonumber \\   \forall (x, \nu) \in [0, 1] \times \mathbb{R}\,. &\label{eq:kernel}
    \end{align}
However, $u(0,t)$ is not constant and \eqref{eq:ut_def}, \eqref{eq:u1_def} is not mapped into $w_t=w_x, w(1,t)=0$ but into a  complex nonlinear PIDE
  \begin{eqnarray}
   w_t(x, t) &=& w_x(x, t) -w_x(0, t) \Omega(x, t) \label{eq:wt_def_exact_k},\\
      w(1, t) &=& 0, \label{eq:w1_def_exact_k}
 \end{eqnarray}
 where
\begin{eqnarray}
\label{eq:omega_def_exact_k}   \Omega(x, t) &=& \int_{0}^{x}
   \left[
   k_{\nu}(x-y, w(0, t))
   \right] 
\, \Big[w(y, t) 
   \nonumber \\&&  
+ \int_{0}^{y} l(y-s, w(0, t))w(s, t)ds\Big]dy
\\
l(x, \nu) &=& 
 \int_{0}^{x} k (x-y, \nu)l(y, \nu)dy 
\nonumber \\ && 
+k(x,\nu) 
, \quad
(x, \nu) \in [0, 1] \times \mathbb{R} \,. \label{eq:l_def_exact_k}
 \end{eqnarray}

The system \eqref{eq:wt_def_exact_k}, \eqref{eq:w1_def_exact_k}, \eqref{eq:omega_def_exact_k}, \eqref{eq:l_def_exact_k} is called the ``target system.'' 
Since the linearization of this 
nonlinear PIDE 
is the transport PDE $w_t=w_x, w(1,t)=0$, a lengthy Lyapunov analysis can be conducted which proves that the equilibrium of this target system is locally exponentially stable in the spatial $H^1$ norm. 

In this paper, we apply an approximation of the GS controller \eqref{eq:U_def} using deep learning. 
The result is a ``perturbed target system,'' which is  even more complex than \eqref{eq:wt_def_exact_k}, \eqref{eq:w1_def_exact_k}, \eqref{eq:omega_def_exact_k}, \eqref{eq:l_def_exact_k}, and requires an even more complicated Lyapunov analysis. 


\section{Existence and Regularity of the "Gain-Schedulable" Backstepping Kernel \texorpdfstring{$k(x,\nu)$}{k(x, v)}}\label{sec:existenceofk}


    \begin{lemma}\label{lemma_k_knu}
    {\em [Existence and  bound for kernel and its derivatives]}
      For each $\beta \in \mathcal{C}^1([0, 1] \times \mathbb{R})$,  the Volterra equation \eqref{eq:kernel}, 
      \begin{equation}
        k(x, \nu) = -\beta(x, \nu) + \int_{0}^{x} \beta(x-y, \nu) k(y, \nu) dy 
      \end{equation}
      has a unique $\mathcal{C}^1([0, 1] \times \mathbb{R})$ solution. In addition, when $\beta$ is only defined for $(x, \nu) \in [0, 1] \times [-B_{\nu}, B_{\nu}], \ B_{\nu} > 0$,  the following 
      holds on the same domain,
      \begin{eqnarray}
        |k(x, \nu)| &\leq& B_{\beta}e^{B_\beta x}\,,
        \label{eq:k_bound} \\
        |k_{\nu}(x, \nu)| &\leq& \alpha_{\nu} e^{\alpha_{\nu} x}(1 + \alpha_{\nu} x)\,,
        \label{eq:knu_bound_1}
      \end{eqnarray}
      where $B_\beta := \|\beta\|_{\infty, [0, 1] \times [-B_{\nu}, B_{\nu}]}$
       and $\alpha_{\nu} := {\rm max}(\|\beta\|_{\infty,[0,1] \times [-B_{\nu}, B_{\nu}]}, \|\beta_{\nu}\|_{\infty, [0,1] \times [-B_{\nu}, B_{\nu}]})$.
      If, furthermore, $\beta_{x \nu}$ exists and is continuous, then $k_{x\nu}$ exists and is continuous. When defined on $[0, 1] \times [-B_{\nu}, B_{\nu}]$, the second derivative  satisfies,        $\forall (x, \nu) \in [0,1] \times [-B_{\nu}, B_{\nu}]$,
\begin{align}
        |k_{x \nu}(x, \nu)| \leq& \, \alpha e^{\alpha x}(1 + 2 \alpha) + \alpha^2 x e^{\alpha x} (\alpha + 1),\label{eq:k_nu_x_bound} \\
 \alpha :=& 
\left\| \max 
\big(|\beta|,|\beta_{\nu}|,|\beta_{x \nu}|,|\beta_x|\big)
\right\|_{\infty, [0, 1] \times [-B_{\nu}, B_{\nu}]}.
      \end{align}
    \end{lemma}

   \begin{proof}
     \underline{Existence and continuity of $k$:}
    Let's recall the results proved in \cite{bhan_neural_2023}. For each $\beta \in \mathcal{C}^0([0, 1], \mathbb{R})$, the equation \begin{equation}
      k(x) = -\beta(x) + \int_{0}^{x} \beta(x-y) k(y)dy, \qquad x \in [0, 1] \,,  \label{eq:kernel_nu}
  \end{equation}
  has a unique $\mathcal{C}^0([0, 1], \mathbb{R})$ solution, which satisfies 
\begin{eqnarray}
\label{k_inequality} 
k &=& \sum_{n=0}^{\infty} \Delta k^n \label{eq:k_delta_kn}
\end{eqnarray}
 and $ |k(x)| \leq \bar{\beta} e^{\bar{\beta}x}$ for $ x \in [0, 1]$,   where
  \begin{eqnarray}
    \Delta k^{n+1}(x) := \int_{0}^{x} \beta(x-y) \Delta k^n(y)dy\,, \quad n\geq0 
  \end{eqnarray}
with $    \Delta k^{0} := -\beta  $ and $
    \bar{\beta} :=  \|\beta\|_{\infty, [0, 1]}$,  and $\Delta k^n$ satisfies 
  \begin{equation}
    \label{eq:delta_kn_bound_1}
    |\Delta k^n(x)| \leq \frac{\bar{\beta}^{n+1}x^n}{n!}\,.
  \end{equation}
For $\nu \in \mathbb{R}$, since $\beta(\cdot, \nu)$ is continuous the previous statements ensure that there exists a continuous kernel $k(\cdot, \nu)$ that satisfies \eqref{eq:kernel_nu}. Introducing the sequence
    \begin{eqnarray}
      \Delta k^0 &:=& - \beta \label{eq:delta_k0}\,,  \\
      \Delta k^{n+1} &:=& \beta * \Delta k^n , \qquad x \in [0, 1]  \label{eq:delta_kn}\,, 
    \end{eqnarray}
    it follows from \eqref{eq:k_delta_kn} that for a given $\nu \in \mathbb{R}$, 
    \begin{eqnarray}
      k(x, \nu) &=& \sum_{n=0}^{\infty} \Delta k^n(x, \nu)\,, \qquad x \in [0, 1]\,,  \label{eq:k_sum} \\
      |\Delta k^n(x, \nu)| &\leq& \frac{\bar{\beta_{\nu}}^{n+1}x^n}{n!}\,, \qquad x \in [0, 1]\,,  \label{eq:delta_kn_bound} 
    \end{eqnarray}
    where $\bar{\beta}_{\nu} := \|\beta(\cdot, \nu)\|_{\infty}$.
    Also, since $\beta$ is continuous,  \eqref{eq:k_sum}, \eqref{eq:delta_kn_bound} give the continuity of $k$. With \eqref{eq:delta_kn_bound} and \eqref{eq:k_sum} we get \eqref{eq:k_bound}. 
    
    Now let's prove that $k$ is $\mathcal{C}^1$ by proving that both its partial derivatives exist and are continuous. To do so we prove that $\sum_{n=0}^{\infty} \Delta k^n_x, \sum_{n=0}^{\infty} \Delta k^n_{\nu}$ converge uniformly on $[0, 1] \times [-B_{\nu}, B_{\nu}]$ for all $B_\nu > 0$.
 
    \underline{Existence and continuity of $k_x$}:
For $n \in \mathbb{N}$, differentiating \eqref{eq:delta_kn} with respect to $x$, since $\beta$ is $\mathcal{C}^1$, gives the existence and continuity of $\Delta k^{n}_x, n \geq 0$ (for $\Delta k^0$ we use \eqref{eq:delta_k0}). The details leading to this conclusion entail the study of 
    \begin{eqnarray}
          \Delta k^{n+1}_x(x, \nu) &=&  \beta(0, \nu) \Delta k^n(x, \nu) \nonumber  + \\ && \int_{0}^{x}\beta_x(x-y, \nu)\Delta k^n(y, \nu)dy
      \label{eq:delta_kn_x}
    \end{eqnarray} 
along with 
\eqref{eq:delta_kn_bound}, which leads to the upper bound
    \begin{equation}
      |\Delta k^{n+1}_x(x, \nu)| \leq \frac{\bar{\beta}_{\nu}^{n+2}x^n}{n!} + \alpha_x \frac{\bar{\beta}_{\nu}^{n+1}x^{n+1}}{(n+1)!} ,
      \label{eq:kn_x_upper}
    \end{equation}
    where $\alpha_x := \|\beta_x\|_{\infty, [0,1] \times [-B_{\nu}, B_{\nu}]}< \infty$ because $\beta$ is assumed to be $\mathcal{C}^1$. 
   The bound provided by \eqref{eq:kn_x_upper} ensures that $\sum_{n=0}^{\infty} \Delta k^n_x$ converges uniformly on $[0, 1] \times [-B_{\nu}, B_{\nu}]$. It follows that $k_x$ exists and is continuous on $[0, 1] \times \mathbb{R}$.

   \underline{Existence and continuity of $k_{\nu}$}: 
    We prove by induction that for all $n \in \mathbb{N}, \Delta k^{n}$ is differentiable with respect to  $\nu$, $\Delta k^{n}_{\nu}$ is continuous and that \begin{equation}
      |\Delta k^{n}_{\nu}(x, \nu)| \leq \frac{(n+1)\alpha_{\nu}^{n+1}x^{n}}{n!}
      \label{eq:kn_mu_bound} \,, 
    \end{equation}
    where $\alpha_{\nu} := {\rm max} \left(\|\beta\|_{\infty}, \|\beta_{\nu}\|_{\infty,[0,1]\times[-B_{\nu}, B_{\nu}]}\right) < \infty$.
    Since $\beta$ is $\mathcal{C}^1$ we have that: 
    \begin{equation}
      |\Delta k^0_{\nu}| \leq \alpha_{\nu} 
      \label{eqn:delta_k0_nu}\,.
    \end{equation}
    Let $n \in \mathbb{N}$ such that the previous statement holds. 
    Taking the derivative of \eqref{eq:delta_kn} with respect to  $\nu$ gives that,  for $(x, \nu) \in [0, 1] \times [-B_{\nu}, B_{\nu}]$, 
    \begin{eqnarray}
      \Delta k^{n+1}_{\nu}(x, \nu) &=&  \int_{0}^{x}(\beta_{\nu}(x-y, \nu)\Delta k^n(y, \nu) \nonumber  \\ &&  + \beta(x-y, \nu) \Delta k^n_{\nu}(y, \nu))dy
      \label{eq:iteration_knu}
    \end{eqnarray}
    Using the upper bounds provided by \eqref{eq:kn_mu_bound} and \eqref{eq:delta_kn_bound} leads to the following inequality:
    \begin{equation}
      |\Delta k^{n+1}_{\nu}(x, \nu)| \leq \frac{(n+2)\alpha_{\nu}^{n+2}x^{n+1}}{(n+1)!}
      \label{eq:kn+1_nu} \,.
    \end{equation}
    The bound provided by \eqref{eq:kn_mu_bound} ensures that $\sum_{n=0}^{\infty} \Delta k^{n}_{\nu}$ converges uniformly on $[0, 1] \times [-B_{\nu}, B_{\nu}]$, so \eqref{eq:knu_bound_1} follows.

    \underline{Existence and continuity of $k_{x \nu}$:}
    If we also assume that $\beta_{x \nu}$ is defined and continuous, we prove that $k_{x \nu}$ exists, and is continuous, by proving that $\sum_{n=0}^{\infty} \Delta k^{n}_{x \nu}$ uniformly converges on $[0, 1] \times [-B_{\nu}, B_{\nu}]$. Taking the derivative of \eqref{eq:iteration_knu} with respect to  $x$ gives that
    \begin{eqnarray}
      \Delta k^{n+1}_{x \nu}(x, \nu) &=& \beta_{\nu}(0, \nu) \Delta k^n(x, \nu) + \beta(0, \nu) \Delta k^n_{\nu}(x, \nu) \nonumber \\
    &&+ \int_0^x (\beta_{x \nu}(x-y, \nu)\Delta k^n(y, \nu) \nonumber
    \\&&+ \beta_{x}(x-y, \nu)\Delta k^n_{\nu}(y, \nu))dy
      \label{eq:kn_nu_x}\,.
    \end{eqnarray}
    Introducing $\alpha := {\rm max} (\alpha_x, \alpha_{\nu}, \|\beta_{x \nu}\|_{\infty, [0, 1] \times [-B_{\nu}, B_{\nu}]})$, using \eqref{eq:kn_mu_bound} and \eqref{eq:delta_kn_bound}, \eqref{eq:kn_nu_x} is bounded by
    \begin{eqnarray}
      |\Delta k^{n+1}_{x \nu} (x, \nu)| \nonumber &\leq& \frac{(n+2)\alpha^{n+2}x^n}{n!} + \frac{(n+2)\alpha^{n+2}x^{n+1}}{(n+1)!}\,, \\ && \qquad (x, \nu) \in [0, 1] \times [-B_{\nu}, B_{\nu}] \,.
      \label{eq:kn_nu_x_bound}
    \end{eqnarray}
    The bound provided by \eqref{eq:kn_nu_x_bound} ensures that $\sum_{n=0}^{\infty} \Delta k^{n}_{x \nu}(x, \nu)$ converges uniformly on $[0, 1] \times [-B_{\nu}, B_{\nu}]$ and leads us to the upper bound \eqref{eq:k_nu_x_bound}.
   \end{proof}

  \section{Neural Operators Approximate Gain-Scheduling Kernel \texorpdfstring{$k$}{k}}\label{sec:approximatek}
  
Analytically solving \eqref{eq:kernel} for arbitrary functions $\beta$, whose second argument is changing in real time, is impossible. We replace the exact solution $k$ with a NO-approximated function $\hat{k}$. The question remains whether the control  \eqref{eq:U_def}, using the approximation $\hat{k}$ instead of the exact kernel $k$, retains  local exponential stability for the equilibrium $u \equiv 0$. We answer affirmatively as long as the approximation is 'close enough.' 
Towards that end, we need an approximation theorem.

%
\begin{theorem}[DeepOnet universal approximation theorem \cite{lu2021advectionDeepONet}]\label{theorem:DeepOnet}
    Let $X \subset \mathbb{R}^{d_x}$ and $Y \subset \mathbb{R}^{d_y}$ be compact sets of vectors $x \in X$ and $y \in Y$, respectively. Let $\mathcal{U}: X \to U \subset \mathbb{R}^{d_u}$ and $\mathcal{V}: Y \to V \subset \mathbb{R}^{d_v}$ be sets of continuous functions $u(x)$ and $v(y)$, respectively. Let $\mathcal{U}$ also be compact. Assume the operator $\mathcal{G}: U \to V$ is continuous. Then, for all $\epsilon > 0$, there exist $m^*, p^* \in \mathbb{N}$ such that for each $m \geq m^*$, $p \geq p^*$, there exist $\theta^{(k)}$, $v^{(k)}$, neural networks $f_N(\cdot; \theta^{(k)})$, $g_N(\cdot; v^{(k)})$, $k = 1, \ldots, p$, and $x_j \in X$, $j = 1, \ldots, m$, with corresponding $\mathbf{u}_m = (u(x_1), u(x_2), \ldots, u(x_m))^T$, such that
    \begin{equation}
      |\mathcal{G}(u)(y) - \mathcal{G}_{\mathbb{N}}(\mathbf{u}_m)(y)| < \epsilon
    \end{equation}
    for all functions $u \in \mathcal{U}$ and all values $y \in Y$ of $\mathcal{G}(u)$.
    Where 
    \begin{equation}
      \mathcal{G}_{\mathbb{N}}(y) = \sum_{k=1}^p g^{\mathcal{N}}(\mathbf{u}_m; v^{(k)})f^{\mathcal{N}}(y; \theta^{(k)})
    \end{equation}
    \end{theorem}

The feedback gain to be scheduled, and approximated, is the output of an operator that we introduce next. To get to the output of that operator, we first need to introduce the set of (admissible) functions at the input of that operator. Based on Theorem \ref{theorem:DeepOnet}, this set of functions needs to be compact. Various bounds on the operator's input functions therefore come into play. We introduce them next. 

Let $B_{\nu} > 0$ and $\breve B = (B_{\beta}, B_{\beta_{\nu}}, B_{\beta_x}, B_{\beta_{x \nu}}) \in (\mathbb{R}^+_*)^{4} $. Let $H$ denote the subset of $\mathcal{C}^1([0, 1] \times [-B_{\nu}, B_{\nu}])$ such that, for all $ \beta \in H$,  
  \begin{itemize}
    \item $\beta, \beta_x, \beta_{\nu}, \beta_{x \nu}$ exist and are Lipschitz with the same Lipschitz constant.
    \item $\|\beta\|_{\infty} < B_{\beta}, \|\beta_{\nu}\|_{\infty} < B_{\beta_{\nu}},\|\beta_x\|_{\infty} < B_{\beta_x},\|\beta_{x \nu}\|_{\infty} < B_{\beta_{x \nu}}$
  \end{itemize}
 
  The reason behind assuming Lipschitzness of $\beta, \beta_x, \beta_{\nu}, \beta_{x \nu}$ is to ensure the compactness of $H$ with the Arzelà-Ascoli theorem, which guarantees  that the set of uniformly bounded and uniformly continuous functions is compact. Lipschitzness is a sufficient condition for uniform continuity and we require this simple property of $\beta$ and its derivatives noted above. 
  
  We endow $H$ with the norm 
  \begin{equation}
      \|\beta\|_H := \|\beta\|_{\infty} + \|\beta_{\nu}\|_{\infty} + \|\beta_x\|_{\infty} + \|\beta_{x \nu}\|_{\infty}\,.
  \end{equation}
  We denote by $\mathcal{K}$ the operator that maps $\beta$ to the kernel $k$ that satifies \eqref{eq:kernel}, i.e., 
  \begin{equation}
      k=\mathcal{K}(\beta)\,.
  \end{equation}
  For the stability analysis explored later in the paper, developing a NO that  approximates $k$ alone is not enough. The perturbed target system obtained with a backstepping transformation with the approximated kernel has terms that contain the derivatives $\frac{\partial}{\partial x}, \frac{\partial}{\partial {\nu}}, \frac{\partial}{\partial {\nu} \partial x}$ of $k$. For this reason, the operator that we approximate with DeepONet must entail more than the operator $\mathcal{K}$. The more elaborate operator that we approximate is denoted by $\mathcal{M}$ and defined next.

  \begin{definition}
  \label{def-operatorK}
  The operator $\mathcal{M} : H \to \mathcal{C}^1([0,1] \times [-B_{\nu}, B_{\nu}], \mathbb{R}) \times \mathcal{C}^0([0,1] \times [-B_{\nu}, B_{\nu}], \mathbb{R})^3$ is defined by 
\begin{equation}
      \mathcal{M}(\beta) = (\mathcal{K}(\beta), \mathcal{K}_1(\beta), \mathcal{K}_{2}(\beta), \mathcal{K}_3(\beta)) \,, 
\end{equation}
and its last three components are defined as
    \begin{eqnarray}
      \mathcal{K}_1(\beta) &=& \frac{\partial}{\partial \nu}\mathcal{K}(\beta)\,,  \\
      \mathcal{K}_2(\beta) &=& \frac{\partial}{\partial \nu}\frac{\partial}{\partial x}\mathcal{K}(\beta)\,,  \\
      \mathcal{K}_3(\beta) &=& \frac{\partial}{\partial x}\mathcal{K}(\beta)\,.\\ \nonumber
    \end{eqnarray}
  \end{definition}

For a neural approximation of $\mathcal{K}$, Theorem \ref{theorem:DeepOnet} relies on the compactness of the input function space of $\mathcal{M}$, as well as on the operator's continuity. In the next lemma, we prove the Lipschitzness of $\mathcal{M}$, and thus its continuity. 

  \begin{lemma}\label{lemma:M_Lipschitzness}
    $\mathcal{M}$ is Lipschitz.
  \end{lemma}

  \begin{proof}
    \underline{Lipschitzness of $\mathcal{K}$:}
     It was shown in \cite{bhan_neural_2023} that
    \begin{equation}
      \|k_1 - k_2\|_{\infty} \leq C \|\beta_1 - \beta_2\|_{\infty},
    \end{equation}
    where $\beta_1, \beta_2 \in \mathcal{C}^0([0, 1])$ such that $\|\beta_1\|_{\infty}, \|\beta_2\|_{\infty} < B$, $C :=e^{3B}$ and $k_1, k_2$ are the solutions of \eqref{eq:kernel_nu}.
    This leads to the Lipschitzness of $\mathcal{K}$ with $e^{3B_{\beta}}$ as the Lipschitz constant.
    \begin{equation}
      \|\mathcal{K}(\beta_1) - \mathcal{K}(\beta_2)\|_{\infty} < C \|\beta_1 - \beta_2\|_{\infty},
      \label{eq:K_Lipschitz}
    \end{equation}
    where $\beta_{1, 2} \in H$. 

    \underline{Lipschitzness of $\mathcal{K}_1$:}
   Let $\beta_1, \beta_2 \in H$ and $k_1=\mathcal{K}(\beta_1), \, k_2=\mathcal{K}(\beta_2)$. From \eqref{eq:kernel} it follows that:
    \begin{eqnarray}
      \delta k_{\nu} &=& -\delta \beta_{\nu} + \partial_{\nu} \beta_1 * \delta k + \delta \beta_{\nu} * k_2 \nonumber   \\ && + \beta_1 * \delta k_{\nu} + \delta \beta * \partial_{\nu} k_2 \label{eq:deltak_nu} \\
      \delta k &=& k_1 - k_2\,, \\
      \delta \beta &=& \beta_1 - \beta_2\,.
    \end{eqnarray}
    Reproducing the successive approximation process, we introduce the sequence $(\delta k_{\nu}^n)$ that satisfies 
    \begin{eqnarray}
      \delta k^{n+1}_{\nu} &=& \beta_1 * \delta k^n_{\nu}, \\
      \delta k^0_{\nu} &=& -\delta \beta_{\nu} + \partial_{\nu} \beta_1 * \delta k + \delta \beta_{\nu} * k_2 \nonumber  \\ &&  + \delta \beta * \partial_{\nu}k_2\,.
    \end{eqnarray}
    Using Lemma \ref{lemma_k_knu} and \eqref{eq:K_Lipschitz} gives the upper bounds
    \begin{eqnarray}
      \|\delta k^0_{\nu}\|_{\infty} &\leq& A \|\delta \beta \|_{H} \,, \\
      |\delta k^n_{\nu}(x, \nu)| &\leq& A \|\delta \beta \|_{H} \frac{B_{\beta}^n x^n}{n!} \,, \\
      A(B_{\beta}, B_{\beta_{\nu}}) &>& 0 \,.
    \end{eqnarray}
    It follows that 
    \begin{eqnarray}
      \delta k_{\nu} &=& \sum_{n=0}^{\infty} \delta k^n_{\nu} \label{eq:delta_k_nu_sum} \,, \\
      \|\delta k_{\nu}\|_{\infty} &\leq& A \|\delta \beta\|_{H} e^{B_{\beta}} \,. \label{eq:delta_k_nu_bound}
    \end{eqnarray}

    \underline{Lipschitzness of $\mathcal{K}_2$:}
     Keeping the same notation as before, differentiating \eqref{eq:kernel} with respect to  $\nu$ and then with respect to  $x$,
    we have that 
    \begin{eqnarray}
      \delta k_{x \nu}(x, \nu) &=& -\delta \beta_{x \nu}(x, \nu) + \partial_\nu \beta_1(0, \nu)\delta k(x, \nu) \nonumber 
      \\ && + \delta \beta_{\nu}(0, \nu)k_2(x, \nu) 
    + \beta_1(0, \nu) \delta k_{\nu}(x, \nu) \nonumber 
    \\ && + \delta \beta(0, \nu) \partial_{\nu} k_2(x, \nu)  + \partial_{\nu} \partial_x \beta_1 * \delta k (x, \nu) \nonumber
    \\ && + \delta \beta_{x \nu} * k_2 (x, \nu)  
   + \partial_x \beta_1(x, \nu) * \delta k_{\nu}(x, \nu) \nonumber 
   \\ && + \delta \beta_x * \partial_{\nu} k_2 (x, \nu)\,, \nonumber 
    \\ && \qquad (x, \nu) \in [0, 1] \times [-B_{\nu}, B_{\nu}] \label{eq:delta_k_x}\,.
    \end{eqnarray}
    Using Lemma \ref{lemma_k_knu}, \eqref{eq:K_Lipschitz} and \eqref{eq:delta_k_nu_bound} gives the upper bound
    \begin{eqnarray}
      \|\delta k_{x \nu}\|_{\infty} &\leq& B \|\delta \beta\|_H \label{eq:kx_Lipschitz}\,, \\
      B(B_{\beta}, B_{\beta_{\nu}}, B_{\beta_x}, B_{\beta_{x \nu}}) &>& 0\,.
    \end{eqnarray}
    
    \underline{Lipschitzness of $\mathcal{K}_3$:}
     Differentiating \eqref{eq:kernel} with respect to  $x$ leads to
    \begin{eqnarray}
      \delta k_x(x, \nu) &=& \beta_1(0, \nu) \delta k (x, \nu) + \delta \beta (0, \nu) k_2(x, \nu) \nonumber \\ &&+ \partial_x \beta_1 * \delta k (x, \nu) + \delta \beta_x * k_2 (x, \nu) \nonumber \\ && - \delta \beta_x(x, \nu)\,.
    \end{eqnarray}
    Using Lemma \ref{lemma_k_knu} and \eqref{eq:K_Lipschitz} gives the upper bound
    \begin{eqnarray}
      \|\delta k_x\|_{\infty} &\leq& D \|\delta \beta\|_H \,, \\
      D(B_\beta, B_{\beta_x}) &>& 0 \,.
    \end{eqnarray}
    The Lipschitzness of $\mathcal{K}, \mathcal{K}_1, \mathcal{K}_2, \mathcal{K}_3$ gives the Lipschitzness of $\mathcal{M}$ (with, for example, a Lipschitz constant that is the maximum of the four Lipschitz constants).
  \end{proof}

  Using Lemma \ref{lemma:M_Lipschitzness} and Theorem \ref{theorem:DeepOnet} we get the following.

\begin{theorem}\label{theorem:NO_k}
  {\em [Existence of a NO to approx. the kernel]}
    For all $ \beta \in H$ and  $\epsilon> 0$, there exists a neural operator $\hat{\mathcal{K}}$ such that, for all $\forall (x, \nu) \in [0, 1] \times [-B_{\nu}, B_{\nu}] $, 
\begin{align}
  |&\mathcal{K}(\beta)(x, \nu) - \hat{\mathcal{K}}(\beta)(x, \nu)| \nonumber \\
      +&|\frac{\partial}{\partial \nu}(\mathcal{K}(\beta) - \hat{\mathcal{K}}(\beta))(x, \nu)| \nonumber \\
      +&|\frac{\partial}{\partial \nu}\frac{\partial}{\partial x}(\mathcal{K}(\beta) - \hat{\mathcal{K}}(\beta))(x, \nu)| \nonumber \\
      +&|\frac{\partial}{\partial x} (\mathcal{K}(\beta) - \hat{\mathcal{K}}(\beta)) (x, \nu)|  < \epsilon 
    \,. \label{eq:khat_property} 
\end{align}
  \end{theorem}

\section{Main Result: Locally Stabilizing DeepONet Gain-Scheduling Feedback}\label{sec:deeponetstabilizes}


\begin{theorem}\label{theorem:loc_stability_base}
  {\em [Loc. stabilization by gain scheduling.]}
  Let $K$, $B_{\nu}$ and  the elements of the vector 
  \begin{equation}
      \breve B = (B_{\beta}, B_{\beta_{\nu}}, B_{\beta_x}, B_{\beta_{x \nu}}), 
  \end{equation} be positive and arbitrarily large. 
 Then for all $c>0$ there exist positive constants $\Omega_0(c, \breve B, B_{\nu})=\mathcal{O}_{c \to \infty}(ce^{-2c}), \epsilon^*(c, \breve B , \Omega_0 ) =\mathcal{O}_{c \to \infty}(e^{-\frac{c}{2}}), f(\breve B),  M(c, \breve B) = f(\breve B)e^{c}$ 
such that
for any $\beta \in \mathcal{C}^1([0, 1] \times \mathbb{R})$ with the properties that $\beta_{x \nu}$ is at least defined on $[0, 1] \times [-B_{\nu}, B_{\nu}]$,
  \begin{eqnarray}
    &&|\beta(x, \nu)| \leq B_{\beta}, \quad |\beta_x(x, \nu)| \leq B_{\beta_x}, \quad |\beta_{\nu}(x, \nu)| \leq B_{\beta_{\nu}},\nonumber  \\ &&|\beta_{x \nu}(x, \nu)| \leq B_{\beta_{x \nu}}, 
    \quad \forall (x, \nu) \in [0, 1] \times [-B_{\nu}, B_{\nu}]\,, 
  \end{eqnarray}
 and that $\beta, \beta_{x}, \beta_{\nu}, \beta_{x \nu}$ are K-Lipschitz,  
 any
 feedback law 
  \begin{equation}
    U(t) = \int_{0}^{1} \hat{k}(1-y, u(0, t))u(y, t)dy,
  \end{equation}
  with $\hat{k}$ being an approximated backstepping kernel provided by Theorem \ref{theorem:NO_k} for any  accuracy $\epsilon \in (0, \epsilon^*)$, guarantees that, 
%
  if the initial condition $u_0 := u(\cdot, 0)$ 
  of the system \eqref{eq:ut_def}, \eqref{eq:u1_def} 
  satisfies 
  \begin{equation}
    \Omega(u_0) \leq \Omega_0\,,
  \end{equation}
  then
  \begin{equation}
    \Omega(u(t)) \leq M \Omega(u_0) e^{-\frac{c}{2}t}\,,\quad \forall t\geq 0\,,
  \end{equation}
where 
\begin{equation} 
  \Omega(u(t)) := u^2(0, t) + \|u(t)\|^2 + \|u_x(t)\|^2
  \,.\label{eq:u_norm_def}
\end{equation}
\end{theorem}

This theorem is proven in Section \ref{sec:H1target}, in which stability is studied for the perturbed target system, and Section \ref{sec:normequiv}, where the norm equivalence is established between the original and target system states, concluding at the end of Section \ref{sec:normequiv}. 


  \section{\texorpdfstring{$H_1$}{H1} Lyapunov Analysis of Perturbed Target System}
\label{sec:H1target}

  For $\beta \in H$, denote $\hat{k} := \hat{\mathcal{K}}(\beta)$ and let $\hat{\mathcal{K}}$ be any NO satisfying Theorem \ref{theorem:NO_k} for any accuracy $\epsilon > 0 $.   We consider now the ``DeepONet-approximated''  transform of \eqref{eq:ut_def}, 
    \begin{eqnarray}
      w(x, t) &=& u(x, t) - \int_{0}^{x} \hat{k}(x-y, u(0, t))u(y, t)dy, \label{eq:w_def} 
    \end{eqnarray}
 which produces the following perturbed target system,
 \begin{eqnarray}
   w_t(x, t) &=& w_x(x, t) -w_x(0, t) \Omega(x, t) \nonumber \\  &&  + [\delta(x, w(0, t)) + \tilde{k}(0, w(0, t) \Omega(x, t))] \nonumber \\ && \times w(0, t) \label{eq:wt_def},\\
      w(1, t) &=& 0, \label{eq:w1_def}
 \end{eqnarray}
 where
 \begin{eqnarray}
   \Omega(x, t) &=& \int_{0}^{x}[k_{\nu}(x-y, w(0, t)) + \delta_2(x-y, w(0, t))] \nonumber \\
                && \times \bigg[w(y, t) \nonumber \\ && + \int_{0}^{y}\hat{l}(y-s, w(0, t))w(s, t)ds\bigg]dy \label{eq:omega_def}\,,  \\
   \delta (x, \nu) &=& -\tilde{k}(x, \nu) + \int_{0}^{x}\beta(x-y, \nu)\tilde{k}(y, \nu)dy\,,  \label{eq:delta_def} \\   
   \delta_2(x, \nu) &=& - \tilde{k}_{\nu}(x, \nu)\,, \\
   \hat{l}(x, \nu) &=& \hat{k}(x, \nu) + \int_{0}^{x}\hat{k}(x-y, \nu) \hat{l}(y, \nu)dy \,, \label{eq:lhat_def} \\
   \tilde{k} &=& k - \hat{k}\,. 
 \end{eqnarray}

\begin{lemma}\label{lemma:inverse_kernel_bound}
 {\em [Upper bound for inverse bkst kernel.]}
    The inverse backstepping kernel $\hat{l}$ introduced as the solution of \eqref{eq:lhat_def} satisfies the inequality
    \begin{eqnarray}
      |\hat{l}(x, \nu)| &\leq& (B_{\beta} + (1+B_\beta)\epsilon)e^{(1+B_{\beta})\epsilon x}\,,  \\
      u(x, t) &=& w(x, t) + \int_0^x \hat{l}(x-y, w(0, t))w(y, t)dy,\label{eq:lhat_ppty}
    \end{eqnarray}
$\forall(x, \nu) \in [0, 1] \times [-B_{\nu}, B_{\nu}]$,    where $u$ is the state of plant \eqref{eq:ut_def}, \eqref{eq:u1_def} and $w$ is the state of the perturbed target system \eqref{eq:wt_def}, \eqref{eq:w1_def}. 
\end{lemma}

\begin{proof}
    First, notice that $\hat{k}$ satisfies
    \begin{equation}
      \hat{k} = -\beta + \beta * \hat{k} + \delta,
      \label{eq:khat_delta}
    \end{equation}
    where $\delta$ is defined in \eqref{eq:delta_def}. 
    Then, reusing the definition of $\hat{l}$ in \eqref{eq:lhat_def}, we have the following equality
    \begin{equation}
      \hat{l} = -\beta + \delta + \delta * \hat{l}\,.
      \label{eq:lhat_delta}
    \end{equation}
    Applying, again, the successive approximation method, it holds that
    \begin{eqnarray}
      |\hat{l}(x, \nu)| &\leq& |B_{\beta} + \bar{\delta}|e^{\bar{\delta} x} \label{eq:lhat_bound}\,, \nonumber \\
      \forall (x, \nu) &\in& [0, 1] \times [-B_{\nu}, B_{\nu}],
    \end{eqnarray}
    where $\|\delta\|_{\infty} \leq\bar{\delta} :=  \epsilon (1+B_{\beta})$. 
    The proof of \eqref{eq:lhat_ppty}  is  done by computing $w + \hat{l} * w$ with $w$ satisfying \eqref{eq:w_def}, and using \eqref{eq:lhat_def}.
\end{proof}

\begin{lemma}
 {\em [Loc. Lyapunov estimate for perturbed target sys.]}
 \label{lemma:V_bound}
 For any $\beta \in H$, for any $c > 0$ 
 there exist positive functions $R_0(\breve B, c, B_{\nu}) = \mathcal{O}_{c \to \infty}(ce^{-c}), \ \epsilon^*(\breve B, c) =\mathcal{O}_{c \to \infty}(e^{-\frac{c}{2}})$ with a decreasing dependence on argument $\breve B$ argument such that for all $\epsilon \in(0, \epsilon^*)$, the Lyapunov function
 \begin{eqnarray}
   V(t) &=& V_1(t) + V_2(t) \label{eq:V_def}\,, \\
   V_1(t) &=& \frac{1}{2}\int_{0}^{1} e^{cx}w^2(x, t)dx 
   \label{eq:V1_def}\,, \\
   V_2(t) &=& \frac{1}{2}\int_{0}^{1} e^{cx}w_x^2(x, t)dx \,,
   \label{eq:V2_def}
 \end{eqnarray}
   satifies :
   \begin{equation}
 V(0) \leq R_0  \quad \Rightarrow \quad  V(t) \leq V(0) e^{-\frac{c}{2}t}, \qquad t \geq 0\,.
   \end{equation}
 \end{lemma}

\begin{proof}
This proof is the central and longest part of the paper, establishing local exponential stability of the perturbed target system, robust to the specific neural network and the accuracy $\epsilon$ used to approximate the gain-scheduled control law. 


Using  Agmon's inequality on $w$ satisfying \eqref{eq:wt_def}, we have 
 \begin{equation}
   |u(0, t)|^2 = |w(0, t)|^2 \leq 2\|w(t)\|\, \|w_x(t)\| \leq 2V(t)\,.
   \label{eq:u0_bound_lyap}
 \end{equation}
Then, if we chose the initial conditions small so that $V(0) \leq \frac{B_{\nu}^2}{2}$ and prove that $V$ is decreasing, the property $|u(0, t)| \leq B_{\nu}$, important throughout our analysis, is enforced $\forall t \geq 0$. 

 Before starting the heavy computation of $\dot{V}$ we  recall from Theorem \ref{theorem:NO_k} and Lemma \ref{lemma:inverse_kernel_bound} several upper bounds, and introduce additional ones, which that are used in the computation, which hold as a result of $|w(0, t)| \leq B_{\nu}$ holding:
 \begin{eqnarray}
   \|\delta(\cdot, w(0, t))\|_{\infty} &\leq& \epsilon (1 + B_{\beta}) =: \bar{\delta} \label{eq:delta_bound}\,,\\
   \|\delta_x(\cdot, w(0, t))\|_{\infty} &\leq& \epsilon (1 + B_{\beta} + B_{\beta_x}) =: \bar{\delta_x} \label{eq:delta_x_bound}\,,\\
   \|\delta_2(\cdot, w(0, t))\|_{\infty} &\leq& \epsilon \label{eq:delta2_lyap}\,,\\
   \|\delta_{2x}(\cdot, w(0, t))\|_{\infty} &\leq& \epsilon \label{eq:delta2_x_lyap}\,,\\
   \|\tilde{k}(\cdot, w(0, t))\|_{\infty} &\leq& \epsilon \label{eq:tilde_k_bound}\,,\\
   \|\hat{l}(\cdot, w(0, t))\|_{\infty} &\leq&  (B_{\beta}+\bar{\delta})e^{\bar{\delta}} := \bar{l} \label{eq:knux_lyap}\,, \\
 \|k_{\nu}(\cdot, w(0, t))\|_{\infty} &\leq& \alpha e^{\alpha}(1+\alpha) =: \bar{k}_{\nu} \,, \label{eq:k_nu_lyap}\\
 \|k_{x \nu}(\cdot, w(0, t))\|_{\infty} &\leq& \alpha e^\alpha (1+2\alpha) \nonumber \\ && + \alpha^2 e^{\alpha} (\alpha + 1) =: \bar{k}_{x \nu},
 \end{eqnarray}
 where 
 \begin{equation}
   \alpha := {\rm max}(B_\beta, B_{\beta_x}, B_{\beta_{\nu}}, B_{\beta_{x \nu}})\,.
  \label{eq:def_alpha}
 \end{equation}
With these inequalities we have that
 \begin{eqnarray}
   \|\Omega(\cdot, t)\|_{\infty} &\leq& \bar{\Omega} \int_{0}^{1}|w(x, t)|dx, \label{eq:omega_bound}  \\
   |\Omega_x(x, t)| &\leq& \bar{\Omega}_{x1}\int_{0}^{1}|w(y, t)|dy + \bar{\Omega}_{x2}|w(x, t)| \label{eq:omega_x_bound},
 \end{eqnarray}
 where 
 \begin{eqnarray}
   \bar{\Omega} &:=& (\bar{k}_{\nu} + \epsilon)(\bar{l} + 1)\,,  \\
   \bar{\Omega}_{x1} &:=& \bar{l}(\bar{k}_{\nu} + \epsilon) + (1 + \bar{l})(\bar{k}_{x \nu} + \epsilon)\,, \\
   \bar{\Omega}_{x2} &:=& \bar{k}_{\nu} + \epsilon\,.
 \end{eqnarray}
 We prove \eqref{eq:omega_bound} first. We assumed that $|w(0, t)| \leq B_{\nu}$, then using triangular inequality, \eqref{eq:k_nu_lyap}, \eqref{eq:delta2_lyap}, \eqref{eq:lhat_bound} we have that
 \begin{eqnarray}
    |\Omega(x, t)| &\leq& (\bar{k}_{\nu} + \epsilon) \bigg(\int_0^x |w(y, t)|dy \nonumber \\ && + \bar{l} \int_0^x \int_0^y |w(s, t)|ds dy\bigg)
 \end{eqnarray}
 and just upper-bound it with $x, y = 1$.
Turning to \eqref{eq:omega_x_bound}, we first take the derivative of \eqref{eq:omega_def} with respect to $x$, 
\begin{eqnarray}
  \Omega_x(x, t) &=& [k_{\nu}(0, w(0, t)) + \delta_2(0, w(0, t))](w(x, t) \nonumber \\ && + \int_{0}^{x} \hat{l}(x-s, w(0, t)) w(s, t) ds) \nonumber \\
                 &&+ \int_{0}^{x} (k_{x \nu}(x-y, w(0, t)) \nonumber \\ && + \delta_{2x}(x-y, w(0, t))) \times \bigg[w(y, t) \nonumber \\ && + \int_{0}^{y}\hat{l}(y-s, w(0, t))w(s, t)ds\bigg]dy\,.\label{eq:omega_x_def}
\end{eqnarray}
We arrive at \eqref{eq:omega_x_bound} assuming that $|w(0, t)| \leq B_{\nu}$ and using \eqref{eq:k_nu_lyap}, \eqref{eq:delta2_lyap}, \eqref{eq:lhat_bound}, \eqref{eq:knux_lyap} and \eqref{eq:delta2_x_lyap}. 
 All of other upper bounds are obtained using Theorem \ref{theorem:NO_k} and Lemmas \ref{lemma_k_knu} and \ref{lemma:inverse_kernel_bound}.

 \underline{Estimate of $V_1$:}
 Taking the derivative of \eqref{eq:V1_def}, leads to the following equalities:
 \begin{eqnarray}
   \dot{V}_1 &=& I_{11} + I_{12} + I_{13} + I_{14}, \label{eq:V1_dot} 
  \end{eqnarray}
where
 \begin{eqnarray}
   I_{11}(t) &=& -\frac{w^2(0, t)}{2} - \frac{c}{2} \int_{0}^{1}e^{cx}w^2(x, t)dx \label{eq:I_11}\,, \\
 I_{12}(t) &=& -w_x(0, t)\int_{0}^{1}e^{cx}w(x, t)\Omega(x, t)dx \label{eq:I_12}\,,\\
I_{13}(t) &=& w(0, t)\int_{0}^{1}e^{cx}w(x, t)\delta(x, t)dx \label{eq:I_13}\,, \\
   I_{14}(t) &=& w(0, t) \tilde{k}(0, w(0, t)) \nonumber \\&& \times \int_{0}^{1}e^{cx}w(x, t)\Omega(x, t)dx \label{eq:I_14}\,.
 \end{eqnarray}
We first work on \eqref{eq:I_12}, $I_{12}(t)$.
Using Young's inequality and \eqref{eq:omega_bound} we arrive at the following inequality
\begin{equation}
  I_{12}(t) \leq \frac{w_x(0,t)^2}{8} + 4 e^{c}\bar{\Omega}^2 V_1(t) \|w(t)\|^2
  \label{eq:I_12_bound} \,.
\end{equation}
Continuing with $I_{13}(t)$ in \eqref{eq:I_13}, using \eqref{eq:delta_bound}, Young's inequality and Cauchy-Scwarz inequalities we have the following upper bound
\begin{equation}
  I_{13}(t) \leq \frac{w(0, t)^2}{8} + 4 \bar{\delta}^2 V_1(t) \frac{(e^{c} - 1)\,. 
}{c}  \label{eq:I_13_bound}
\end{equation}
Moving to $I_{14}(t)$ in \eqref{eq:I_14}, Young's inequality and \eqref{eq:omega_bound} give the following inequality,
\begin{equation}
  I_{14}(t) \leq \frac{w(0, t)^2}{8} + 4 \epsilon^2 \bar{\Omega}^2 e^{c} \|w(t)\|^2 V_1(t)\,. 
  \label{eq:I_14_bound}
\end{equation}
 Gathering \eqref{eq:I_11}, \eqref{eq:I_12_bound}, \eqref{eq:I_13_bound}, \eqref{eq:I_14_bound} we arrive at the following inequality
\begin{eqnarray}
  \dot{V_1}(t) &\leq& \frac{w^2_x(0, t)}{8} - \frac{w^2(0, t)}{4} \nonumber \\
      &&-V_1(t)[c - 4\epsilon^2 \bar{\Omega}^2 e^{c} \|w(t)\|^2 \nonumber \\ && - 4\bar{\delta}^2 \frac{e^{c} - 1}{c} - 4 e^{c} \bar{\Omega}^2 \|w(t)\|^2] \label{eq:V_1_bound}\,.
\end{eqnarray}

\underline{Estimate of $V_2$:}
Note that differentiating \eqref{eq:w_def} with respect to  $x$ leads to the following PIDE system satisfied by $w_x$:
\begin{eqnarray}
  w_{xt}(x, t) &=& w_{xx}(x, t) - w_x(0, t)\Omega_x(x, t) \nonumber \\
               &&+ \bigg(\delta_x(x, w(0, t)) \nonumber \\ &&+ \tilde{k}(0, w(0, t) \Omega_x(x, t))\bigg)w(0, t)\,,\\
  w_x(1, t) &=& \Omega(1, t)w_x(0, t) - w(0, t)[\delta(1, w(0, t)) \nonumber  \\ && + \tilde{k}(0, w(0, t))\Omega(1, t)] \,.\label{eq:wx_1}
\end{eqnarray}
 Taking the derivative of \eqref{eq:V2_def} leads to the following:
 \begin{eqnarray}
   \dot{V}_2 &=& I_{21}+ I_{22} + I_{23} + I_{24} \label{eq:V2_dot}\,, \\
   I_{21}(t) &=& \frac{e^{c}}{2} w^2_x(1, t) - \frac{w^2_x(0, t)}{2} \nonumber \\ && - \frac{c}{2} \int_{0}^{1}e^{cx}w_x^2(x, t)dx \label{eq:I_21}\,, \\
   I_{22}(t) &=& -w_x(0, t)\int_{0}^{1}e^{cx}w_x(x, t)\Omega_x(x, t)dx \label{eq:I_22} \,,\\
   I_{23}(t) &=& w(0, t)\int_{0}^{1} e^{cx}w_x(w, t) \delta_x(x, w(0, t))dx \label{eq:I_23}\,, \\
   I_{24}(t) &=& w(0, t)\tilde{k}(0, w(0, t)) \nonumber \\ && \times \int_{0}^{1} e^{cx}w_x(w, t) \Omega_x(x, t)dx \label{eq:I_24}\,.
 \end{eqnarray}
 Starting with $I_{21}(t)$. We'll specifically bound the term $w_x^2(1, t)$. 
 But first let's notice that thanks to Hölder's inequality we have the upper bound
 \begin{equation}
   \int_{0}^{1} |w(x, t)|dx \leq \|w(t)\| \leq \sqrt{2V(t)} \,.
  \label{eq:w_l1_bound}
 \end{equation}
 Using \eqref{eq:wx_1}, \eqref{eq:omega_bound}, \eqref{eq:delta_bound}, \eqref{eq:tilde_k_bound} and \eqref{eq:w_l1_bound} we have the following inequality for $w_x^2(1, t)$
 \begin{eqnarray}
   w_x^2(1, t) &\leq& 2w_x^2(0, t) \bar{\Omega}^2 V(t) \nonumber \\ && + w^2(0, t) \epsilon^2 (1 + B_{\beta} + \sqrt{2V(t)} \bar{\Omega} )^2 + \epsilon^2 w^2(0, t) \nonumber \\
   &&+ 2V(t) w^2_x(0, t) \bar{\Omega}^2 (1 + B_{\beta} + \sqrt{2V(t)}\bar{\Omega})^2  \label{eq:wx_1_bound}\,.
 \end{eqnarray}
 Then using \eqref{eq:I_21}, \eqref{eq:wx_1_bound} gives
 \begin{eqnarray}
   I_{21}(t) &\leq& \frac{e^{c}}{2}[2w_x^2(0, t) \bar{\Omega}^2 V(t) \nonumber \\ && + w^2(0, t) \epsilon^2 (1 + B_{\beta} + \sqrt{2V(t)} \bar{\Omega} )^2 + \epsilon^2 w^2(0, t) \nonumber \\
          &&+ 2V(t) w^2_x(0, t) \bar{\Omega}^2 (1 + B_{\beta} + \sqrt{2V(t)}\bar{\Omega})^2] \nonumber \\ && -\frac{w^2_x(0, t)}{2} - c V_2(t) \label{eq:I_21_bound}\,.
 \end{eqnarray}
Next, we move on to $I_{22}(t)$ in \eqref{eq:I_22}. Using \eqref{eq:omega_x_bound} as well as Young's inequality and Cauchy-Schwarz inequality we have that 
 \begin{align}
   I_{22}(t) \leq \frac{w_x(0, t)^2}{4} + 4 e^{c}V_2(t)\|w(t)\|^2(\bar{\Omega}_{x1}^2 + \bar{\Omega}_{x2}^2)\,.
  \label{eq:I_22_bound}
 \end{align}
Advancing to $I_{23}(t)$ in \eqref{eq:I_23}, using Young's inequality, \eqref{eq:delta_x_bound} and Cauchy-Schwarz inequality we obtain
\begin{align}
  I_{23}(t) \leq \frac{w^2(0, t)}{8} + 4 V_2(t) \bar{\delta_x}^2 \frac{(e^{c} - 1)}{c} \,. \label{eq:I_23_bound}
 \end{align}
 Finally for $I_{24}(t)$ in \eqref{eq:I_24}, using Young's inequality, \eqref{eq:omega_x_bound}, \eqref{eq:tilde_k_bound} and the Cauchy-Schwarz inequality leads to 
 \begin{align}
   I_{24}(t) \leq \epsilon^2 \frac{w^2(0, t)}{4} + 4e^{c}V_2(t)\|w(t)\|^2(\bar{\Omega}_{x1}^2 + \bar{\Omega}_{x2}^2) \label{eq:I_24_bound} \,. 
 \end{align}
 Using the inequalities \eqref{eq:I_21_bound}, \eqref{eq:I_22_bound}, \eqref{eq:I_23_bound} and \eqref{eq:I_24_bound} gives 
 \begin{eqnarray}
   \dot{V_2}(t) &\leq& w^2(0, t)\bigg[\frac{\epsilon^2}{4} + \frac{1}{8} \nonumber \\ && + \epsilon^2 \frac{e^{c}}{2}(1 + (1+ B_{\beta} + \sqrt{2V(t)}\bar{\Omega})^2)\bigg] \nonumber \\
             &&- w_x^2(0, t)\bigg[\frac{1}{4} - e^{c}V(t) \nonumber \\ && \times \bar{\Omega}^2(1+(1+B_{\beta}+\sqrt{2V(t)}\bar{\Omega})^2))\bigg] \nonumber \\
             &&-V_2(t) \bigg[c - 8e^{c}(\bar{\Omega}_{x1}^2 + \bar{\Omega}_{x2}^2)\|w(t)\|^2 \nonumber \\ &&- 4 \bar{\delta}_x^2 \frac{(e^{c} - 1)}{c} \bigg] \label{eq:V_2_bound}\,. 
 \end{eqnarray} 
 In summary, gathering \eqref{eq:V_1_bound} and \eqref{eq:V_2_bound} we have 
 \begin{eqnarray}
   \dot{V}(t) &\leq& -w(0, t)^2\bigg[\frac{1}{8} - \frac{\epsilon^2}{4} \nonumber \\ && - \frac{\epsilon^2e^{c}}{2}(1 + (1 + B_{\beta} + \sqrt{2V(t)}\bar{\Omega})^2)\bigg] \label{eq:w0_bound}\nonumber  \\
              &&-w_x(0, t)^2\bigg[\frac{1}{8} - e^{c}V(t) \nonumber \\ &&\nonumber  \times  \bar{\Omega}^2(1 + (1+B_{\beta} +\sqrt{2V(t)}\bar{\Omega})^2)\bigg] \label{eq:wx0_bound} \\
              &&-V_1(t)\bigg[c -4\epsilon^2 \bar{\Omega}^2e^{c}\|w(t)\|^2 \nonumber \\ && - 4\bar{\delta}^2 \frac{e^{c}-1}{c}  - 4e^{c}\bar{\Omega}^2\|w(t)\|^2\bigg] \label{eq:V1_final_bound} \nonumber \\
              &&-V_2(t) \bigg[c -8e^{c}(\bar{\Omega}_{x1}^2 + \bar{\Omega}_{x2}^2)\|w(t)\|^2 \nonumber \\ && -4\bar{\delta}_x^2 \frac{e^{c} - 1}{c} \bigg] \label{eq:V2_final_bound}\,.
 \end{eqnarray}
We kept using $\bar{\Omega}, \bar{\Omega}_{x1}, \bar{\Omega}_{x2}$ without specifying that they are functions of $\epsilon$. In what follows we restrict, for simplicity, the choice of $\epsilon$ to  $\epsilon \in (0, 1]$ and use the notation $\bar{\Omega} := \bar{\Omega}(1), \bar{\Omega}_{x1} := \bar{\Omega}_{x1}(1), \bar{\Omega}_{x2} := \bar{\Omega}_{x2}(1)$. We emphasize that all the previous inequalities are  valid for all $\epsilon \leq 1$.
We introduce the $\mathcal{K}_{\infty}$ function 
\begin{equation}
  \beta_2(V) := e^{c} \bar{\Omega}^2V + e^{c} V \bar{\Omega}^2(1 + B_{\beta} + \sqrt{2V}\bar{\Omega})^2 \,, 
  \label{eq:beta_2_def}
\end{equation}
and the quantity 
$    R_1 := \beta_2^{-1}\left({1}/{8}\right) > 0$. 
 That way, if we assume that $V(t) \leq R_1$, then \eqref{eq:wx0_bound} is negative.
 Moving on to \eqref{eq:w0_bound}, we introduce 
 \begin{equation}
   \epsilon_1 := {\rm min}\left(1, \frac{1}{\sqrt{2 + 4e^{c}(1 + (1 + B_{\beta} + \sqrt{2R_1}\bar{\Omega})^2)}}\right)
   \label{eq:eps_1}\,. 
 \end{equation}
If we take $\epsilon \leq \epsilon_1$, and $V \leq R_1$ we  have that \eqref{eq:w0_bound} is negative. Advancing to \eqref{eq:V1_final_bound}, we introduce the positive constants
 \begin{align}
   R_2 :=& {\rm min}\left(R_1, \frac{ce^{-c}}{32\bar{\Omega}^2}\right) > 0 \,, \\
   \epsilon_2 :=& {\rm min}\bigg(\epsilon_1, 
   \sqrt{\frac{\frac{c}{2} - 8e^{c}\bar{\Omega}^2 R_2}{8\bar{\Omega}^2e^{c}R_2 + 4(1+B_{\beta})^2 \times \frac{e^{c}-1}{c}}}\bigg) 
   \,.\label{eq:eps_2}
 \end{align}
With these definitions, if we choose $\epsilon \leq \epsilon_2, V(0) \leq R_2$, we  have \eqref{eq:w0_bound}, \eqref{eq:wx0_bound} negative and \eqref{eq:V1_final_bound} is bounded by $-\frac{c}{2}V_1$. This is achieved by noticing that $\|w(t)\|^2 \leq 2V(t)$. Finally, moving on to \eqref{eq:V2_final_bound}, we introduce
 \begin{eqnarray}
   R_3 &:=& {\rm min}\left(R_2, \frac{ce^{-c}}{64(\bar{\Omega}_{x1}^2 + \bar{\Omega}_{x2}^2)}\right) > 0\,,   \\
   \epsilon^{*} &:=& {\rm min}\bigg(\epsilon_2, \frac{1}{2(1+B_{\beta}+B_{\beta_x})}\nonumber \\ && \times \sqrt{\frac{\frac{c^2}{2} - 16 c e^{c}(\bar{\Omega}_{x1}^2 + \bar{\Omega}_{x2}^2)R_3}{(e^{c}-1)}}\bigg) > 0 \label{eq:eps_star}\,, 
 \end{eqnarray} 
 with which we get that, if we choose $\epsilon \leq \epsilon^{*}, V(0) \leq R_3$, we  have that \eqref{eq:w0_bound}, \eqref{eq:wx0_bound} are negative, \eqref{eq:V1_final_bound} is bounded by $ -\frac{c}{2}V_1$, and \eqref{eq:V2_final_bound} is bounded by $ -\frac{c}{2}V_2$. 
To recapitulate our many majorizations, choosing $  \epsilon \leq \epsilon^*$ and 
\begin{eqnarray}
  V(0) &\leq& {\rm min}\left({R_3, \frac{B_\nu^2}{2}}\right) =: R_0 \label{eq:R0_def}\,. 
\end{eqnarray}
we arrive at the inequality
 \begin{eqnarray}
   \dot{V}(t) &\leq& -\frac{c}{2}V \label{eq:Vdot_bound}, \qquad t\geq 0\,.
 \end{eqnarray}
 Using the comparison lemma we complete the proof.
\end{proof}

\begin{lemma}{\em [Loc. e.s. of perturbed target sys.]}
\label{lemma:local_stability_target}
   For all $ c > 0$, there exist  functions $\epsilon^{*}(\breve B, c) = \mathcal{O}_{c \to \infty}(e^{-\frac{c}{2}}), \quad \Psi_0(B_{\nu}, \breve B, c) = \mathcal{O}_{c \to \infty}(ce^{-2c}) > 0$ with a decreasing dependence in the argument $\breve B$  such that, for any $\beta \in H$ and any NO satisfying Theorem \ref{theorem:NO_k} for that $\beta$ with any $\epsilon \in (0, \epsilon^*)$, where $\epsilon^*$ is defined in \eqref{eq:eps_star}, if the perturbed target system \eqref{eq:wt_def}, \eqref{eq:w1_def}  is initialized with $w_0 := w(\cdot, 0)$ such that $\Psi(w_0) \leq \Psi_0$, then
  \begin{equation}
    \Psi(w(t)) \leq  e^{c} \Psi(w_0) e^{-\frac{c}{2}t}\,, \qquad t \geq 0 \,.
  \end{equation}
  where
  \begin{equation}
  \Psi(w(t))  := \|w(t)\|^2 + \|w_x(t)\|^2\,. 
  \label{eq:def_psi}
\end{equation}
\end{lemma}

\begin{proof}
  Using the Lyapunov function defined in \eqref{eq:V_def} for $c > 0$, Lemma \ref{lemma:V_bound} ensures the existence of $R_0, \epsilon^{*} > 0$ such that if $V(0) \leq R_0, \ \epsilon \leq \epsilon^*$ 
  \begin{equation}
    V(t) \leq V(0) e^{-\frac{c}{2}t}\,. 
  \end{equation}
  Note that 
  \begin{equation}
    \Psi(w(t)) \leq 2V(t) \leq e^{c} \Psi(w(t))\,. 
    \label{eq:lyap_psi_equiv}
  \end{equation}
  Choosing 
  \begin{equation}
    \Psi_0 := 2 e^{-c} R_0\,, 
    \label{eq:Psi_0_def}
  \end{equation}
   and $\epsilon^*$ as in \eqref{eq:eps_star} completes the proof. 
\end{proof}

\section{Local Stability in Original Plant Variable}
\label{sec:normequiv}

\begin{lemma}\label{lemma:equiv_norm_target_base}
  {\em [Equiv. norm perturb. target system]}
  There exist  $\rho, \delta>0$ such that $\forall t \geq 0$, if $|u(0, t)| \leq B_{\nu}$ then,
  \begin{eqnarray}
    \Psi(w(t)) &\leq& \delta \,\Omega(u(t))\,,\label{eq:psi_bound} \\
    \Omega(u(t)) &\leq& \rho\, \Psi(w(t))\,.\label{eq:omega_norm_bound}
  \end{eqnarray}
\end{lemma}

\begin{proof}
   We  use Lemma \ref{lemma_k_knu}, \eqref{eq:k_bound}. We thus have that
  \begin{eqnarray}
    |\hat{k}(x, \nu)| &\leq& \epsilon + B_{\beta}e^{B_{\beta}}  =: \bar{k} \label{eq:khat_bound}\,,  \\
    |\hat{k}_x(x, \nu)| &\leq& \epsilon + (B_{\beta_x} + B_{\beta}^2e^{B_{\beta}}\nonumber  \\ && + B_{\beta_x}B_{\beta}e^{B_{\beta}}) =: \bar{k}_x \label{eq:khat_x_bound}\,, 
\end{eqnarray}
  $\forall (x, \nu) \in [0, 1] \times [-B_{\nu}, B_{\nu}]$, where $\hat{k} := \hat{\mathcal{K}}(\beta)$. 
  With the assumption that 
 $  |u(0, t)|  \leq  B_{\nu} $, 
   we apply the  inequalities \eqref{eq:khat_bound}, \eqref{eq:khat_x_bound} for $\hat{k}(x, u(0, t)), \hat{k}_x(x, u(0, t))$.
  Using \eqref{eq:w_def}, \eqref{eq:khat_bound}, Young's, and the Cauchy-Schwarz inequalities, we have that
  \begin{equation}
    \|w(t)\|^2 \leq 2\|u(t)\|^2 + 2\bar{k}^2 \|u(t)\|^2 \label{eq:w_norm_bound}\,. 
  \end{equation}
  Using \eqref{eq:w_def}, \eqref{eq:khat_x_bound}, \eqref{eq:khat_bound}, the Cauchy-Schwarz, and Young's inequality we have that
  \begin{equation}
    \|w_x(t)\|^2 \leq 4 \|u_x(t)\|^2 + 4 \bar{k}_x^2 \|u(t)\|^2 + 2 \bar{k}^2 \|u(t)\|^2 \,. \label{eq:wx_norm_bound}
  \end{equation}
  Gathering \eqref{eq:w_norm_bound} and \eqref{eq:wx_norm_bound} we have that
\begin{equation}
  \Psi(w(t)) \leq (6 + 4 \bar{k}^2 + 4 \bar{k}_x^2) \Omega(u(t)) =: \delta\,\Omega(u(t))\,.
\end{equation}
Since $u(0, t) = w(0, t)$, we still have $|w(0, t)| \leq B_{\nu}$.
Before constructing $\rho$, let's notice that taking the derivative of \eqref{eq:lhat_def} with respect to $x$ gives
\begin{align}
  \hat{l}_x(x, w(0, t)) =& \hat{k}_x(x, w(0, t)) +  \hat{k}(0, w(0, t)) \hat{l}(x, w(0, t)) \nonumber \\ & + \int_0^x\hat{k}_x(x-y, w(0, t)) 
  \hat{l}(y, w(0, t))dy\,. 
\end{align}
Then using \eqref{eq:khat_bound}, \eqref{eq:khat_x_bound}, \eqref{eq:lhat_bound} we have that
\begin{equation}
  |\hat{l}_x(x, w(0, t))| \leq \bar{k}_x + \bar{k} \bar{l} + \bar{k}_x \bar{l} =: \bar{l}_x\,. 
  \label{eq:lhat_x_bound}
\end{equation}
Let's also notice that 
\begin{equation}
  u^2(0, t) = w^2(0, t) \leq \Psi(w(t))\,. 
  \label{eq:u2_0_bound}
\end{equation}
Using \eqref{eq:u2_0_bound}, \eqref{eq:lhat_bound}, \eqref{eq:lhat_x_bound}, and \eqref{eq:lhat_ppty} it is shown the same way as before that
\begin{equation}
  \Omega(u(t)) \leq (7 + 4 \bar{l}^2 + 4\bar{l}_x^2)\Psi(w(t)) =: \rho \, \Psi(w(t))\,. 
  \label{eq:rho_def}
\end{equation}
The constants $\rho, \delta$  increase in $\epsilon$. Recalling \eqref{eq:eps_1},  $\epsilon \leq \epsilon^* \leq 1$. To remove the dependency  of $\delta, \rho$ on $\epsilon$, we   take their values at $\epsilon=1$ 
and all the previous inequalities remain  valid.
\end{proof}

With the lemmas, we complete the proof of Theorem \ref{theorem:loc_stability_base}.

\begin{proof}[of Theorem \ref{theorem:loc_stability_base}]
  From Lemma \ref{lemma:local_stability_target} there exist $\epsilon^*(\breve B, c), \Psi_0(\breve B, B_{\nu}, c) > 0$, such that if $\Psi(w_0) \leq \Psi_0$, $\epsilon \in (0, \epsilon^*)$ then
  \begin{equation}
    \Psi(w(t)) \leq e^{c} \Psi(w_0) e^{-\frac{c}{2} t}, \qquad t \geq 0\,. 
  \end{equation}
  Recalling \eqref{eq:u0_bound_lyap}, \eqref{eq:R0_def}, \eqref{eq:lyap_psi_equiv} and \eqref{eq:Psi_0_def}
  If $\Psi(w_0) \leq \Psi_0$ then we also have that 
  \begin{equation}
    u^2(0, t) \leq B_{\nu}^2, \qquad \forall t \geq 0\,.  \label{eq:condition}
  \end{equation}
  Then we introduce the quantity $\Omega_0 := {\rm min}(\frac{\Psi_0}{\delta}, B_{\nu}^2)$, where $\delta$ is defined in Lemma \ref{lemma:equiv_norm_target_base}. We choose $\Omega(u_0)$ such that $\Omega(u_0) \leq \Omega_0$.
  Since $u(0, 0)^2 \leq \Omega_0 \leq B_{\nu}^2$, we use  inequality \eqref{eq:psi_bound}, which ensures $\Psi(w_0) \leq \delta \Omega(u_0) \leq \Psi_0$. Since \eqref{eq:condition} is now valid, we use \eqref{eq:omega_norm_bound} which ensures that
  \begin{eqnarray}
    \Omega(u(t)) &\leq& \rho \Psi(w(t)) \leq \rho e^c\Psi(w_0)e^{-c^* t} \nonumber \\
                 &\leq& \rho e^c \delta\Omega(u_0) e^{-c^* t}\,,
  \end{eqnarray}
  with $\rho, \delta$ defined in Lemma \ref{lemma:equiv_norm_target_base}, completing the proof of 
  local exponential stability in $H^1$ asserted in Theorem \ref{theorem:loc_stability_base}.
\end{proof}

\section{A Gain-Only Approach to Approximate Gain Scheduling}\label{sec:gain-only}

In this section we present an alternative approach to the approximation of the kernel and the stability analysis. Instead of approximating the four functions $k, k_{\nu}, k_x, k_{x \nu}$, as per Definition \ref{def-operatorK}  and Lemma \ref{lemma:M_Lipschitzness}, we approximate only the functions $k, k_{\nu}$. This relaxes the DeepONet training requirements. 

This training reduction comes at a price in  analysis. Instead of the approximate  transformation \eqref{eq:w_def}, which employs the kernel $\hat k$, we use the {\em exact}
 transformation \eqref{eq:w_true_k_1} with kernel $k$, which, along with the same feedback law $U$ defined in \eqref{eq:U_def}, maps the system \eqref{eq:ut_def}, \eqref{eq:u1_def} into the nonlinear PIDE
  \begin{eqnarray}
    w_t(x, t) &=& w_x(x, t) -w_x(0, t) \Omega(x, t) ,\label{eq:wt_pbc_def}\\
   w(1, t) &=& \Gamma(t), \label{eq:w_perturbed_bc}
 \end{eqnarray}
 where 
 \begin{eqnarray}
   \Gamma(t) &=& -\int_{0}^{1} \tilde{k}(1-y, w(0, t)) \bigg[w(y, t) \\ && \quad + \int_{0}^{y}l(y-s, w(0, t))w(s, t)ds \bigg]dy, \label{eq:gamma_def}
 \end{eqnarray}
 and for the reader's convenience we recall \eqref{eq:omega_def_exact_k} and \eqref{eq:l_def_exact_k}
 \begin{eqnarray}
 \label{eq:omega_def_exact_k2}
   \Omega(x, t) &=& \int_{0}^{x}
   k_{\nu}(x-y, w(0, t))
   \bigg[w(y, t) \nonumber \\ && + \int_{0}^{y} l(y-s, w(0, t))w(s, t)ds\bigg]dy,
   \\
l(x, \nu) &=& k(x, w(0, t)) \nonumber \\ && + \int_{0}^{x} k (x-y, w(0, t))l(y, w(0, t))dy 
\,. 
 \end{eqnarray}
A close comparison of \eqref{eq:wt_pbc_def} with \eqref{eq:wt_def}, as well as \eqref{eq:omega_def_exact_k2} with \eqref{eq:omega_def}, reveals that the pertubation of the target system in the PDE domain no longer contains a perturbation based on the kernel approximation but retains a perturbation due to the GS-induced nonlinearity, \eqref{eq:wt_pbc_def}. Additionally, by comparing \eqref{eq:w_perturbed_bc} with \eqref{eq:w1_def}, we see that the boundary condition is now perturbed, by $\Gamma$ defined in \eqref{eq:gamma_def}. The consequence of the approximation-based perturbation moving from the PDE domain into the boundary is that the difficulty of the Lypunov analysis increases somewhat. 
A change of the norm is needed because of the perturbation in the boundary condition \eqref{eq:w_perturbed_bc}. We cannot ensure through Agmon's inequality alone that $w(0, t)$ remains bounded since $w(1, t) \neq 0$ which was a simplifying aspect of the  proof of Lemma \ref{lemma:V_bound}. Instead of the norm $\|w(t)\|^2 + \|w_x(t)\|^2 $, $w(t) \in \mathcal{C}^1([0, 1])$, we work with 
 \begin{equation}
 \Psi(w(t)) = w^2(0, t) + \|w(t)\|^2 + \|w_x(t)\|^2 \,. 
  \label{eq:new_psi_def}
 \end{equation}

\subsection{Relaxed kernel approximation}

We  redefine  the set $H$ and the NO used to approximate $k$. We relax $H$  to the subset of $\mathcal{C}^1([0, 1] \times [-B_{\nu}, B_{\nu}])$ such that for all $\beta \in H$,
 \begin{itemize}
    \item $\beta, \beta_x, \beta_{\nu}, \beta_{x \nu}$ exist and $\beta, \beta_x$ are Lipschitz with the same Lipschitz constant,
   \item $\|\beta\|_{\infty} < B_{\beta}, \|\beta_{\nu}\|_{\infty} < B_{\beta_{\nu}}, \|\beta_{x}\|_{\infty} < B_{\beta_{x}}, \|\beta_ {x \nu}\|_{\infty} < B_{\beta_{x \nu}}$.
 \end{itemize}
 Note that with the new $H$, we aren't required to have $\beta_x, \beta_{x \nu}$ Lipschitz. The reason behind this choice is that we don't approximate $k_x, k_{x \nu}$  but only $k, k_{\nu}$. We don't require $\beta_x, \beta_{x \nu}$ uniformly bounded for the approximation but so that feedback  will be available to all functions  $\beta \in H$.
 With this new set $H$ we require a weaker version of Theorem \ref{theorem:NO_k}. 

\begin{theorem}\label{theorem:NO_k_weak}
  {\em [Neural operator to only approximate $k, k_{\nu}$.]}
Given the operator $\mathcal{K}$ 
  \begin{equation}
  \mathcal{K}:  \beta \to k(\beta)\,,
  \end{equation}
  for all $\beta \in H$ and $\epsilon > 0$, there exists a neural operator $\hat{\mathcal{K}}$ such that, $\forall (x, \nu) \in [0, 1] \times [-B_{\nu}, B_{\nu}]$,
  \begin{eqnarray}
  |\mathcal{K}(\beta)(x, \nu) - \hat{\mathcal{K}}(\beta)(x, \nu)| 
  +|\frac{\partial}{\partial \nu}(\mathcal{K} - \hat{\mathcal{K}})(x, \nu)| < \epsilon\,.
  \end{eqnarray}
\end{theorem}

Having introduced the relaxed operator, we  proceed to achieving a stabilization result analogous to Theorem \ref{theorem:loc_stability_base}.

 \begin{theorem}\label{theorem:loc_stability_base_exact_k}
  {\em [Loc. stabilization by gain scheduling.]}
  Let $K, B_{\nu}$ and the elements of the vector 
  \begin{equation}
    \breve B = (B_{\beta}, B_{\beta_{\nu x}}, B_{\beta_x}, B_{\beta_{x \nu}}) \,, 
  \end{equation}
  be positive and arbitrarily large. Then for all $c > 0$ there exist positive constants $\Omega_0(c, \breve B, B_{\nu}) = \mathcal{O}_{c \to \infty}(e^{-2c})$, $\epsilon^*(c, \breve B, \Omega_0)=\mathcal{O}_{c \to \infty}(e^{-\frac{c}{2}}), f(\breve B), M(c, \breve B) = f(\breve B)ce^c$ such that for any $\beta \in \mathcal{C}^1([0, 1] \times \mathbb{R})$ with the properties that $\beta_{x \nu}$ is at least defined on $[0, 1] \times [-B_{\nu}, B_{\nu}]$,
  \begin{eqnarray}
    |\beta(x, \nu)| &\leq& B_{\beta}, \\  |\beta_x(x, \nu)| &\leq& B_{\beta_x}, \\ |\beta_{\nu}(x, \nu)| &\leq& B_{\beta_{\nu}}, \\|\beta_{x \nu}(x, \nu)| &\leq& B_{\beta_{x \nu}} , \\
    \forall (x, \nu) &\in& [0, 1] \times [-B_{\nu}, B_{\nu}]\,, \nonumber
  \end{eqnarray}
 and that $\beta, \beta_{\nu}$ are K-Lipschitz,  
 any
 feedback law 
  \begin{equation}
    U(t) = \int_{0}^{1} \hat{k}(1-y, u(0, t))u(y, t)dy,
  \end{equation}
  with $\hat{k}$ being an approximated backstepping kernel provided by Theorem \ref{theorem:NO_k_weak} for any  accuracy $\epsilon \in (0, \epsilon^*)$, guarantees that, 
  if the initial condition $u_0 := u(\cdot, 0)$ 
  of  system \eqref{eq:ut_def}, \eqref{eq:u1_def} 
  satisfies 
  \begin{equation}
    \Omega(u_0) \leq \Omega_0\,,
  \end{equation}
  then
  \begin{equation}
    \Omega(u(t)) \leq M \Omega(u_0) e^{-\frac{c}{2}t}\,,\quad \forall t\geq 0\,,
  \end{equation}
where 
\begin{equation} 
  \Omega(u(t)) := u^2(0, t) + \|u(t)\|^2 + \|u_x(t)\|^2
  \,.
\end{equation}
\end{theorem}

Notice that we now have a quantitatively slightly weaker result than the one stated in Lemma \ref{theorem:loc_stability_base}: an overshooting coefficient proportional to $ce^c$ instead of $e^c$, and, additionally, the restriction on $\Psi_0$ has changed from $\mathcal{O}_{c \to \infty}(ce^{-2c})$ to the more conservative $\mathcal{O}_{c \to \infty}(e^{-2c})$. This slight weakening of the result follows from not having to approximate $k_{x \nu}, k_{x}$. 

This theorem is proven in Section \ref{subsection:lyap_pbc_analysis}, in which the stability is studied for the perturbed target system, and section \ref{subsection:equiv_perturbed_base} where  the norm equivlence is established between the original and target system states.

\subsection{Lyapunov analysis of target system with perturbed boundary conditions}\label{subsection:lyap_pbc_analysis}

Before Lyapunov analysis, we state bounds for the kernel of the inverse transformation, which is the counterpart of Lemma \ref{lemma:inverse_kernel_bound}, but with the exact inverse backstepping kernel $l$ instead.

\begin{lemma}\label{lemma_l_bound}
{\em [Upper bound the for the exact inverse backstepping kernel and its derivative]}
The inverse backstepping transform's kernel $l$, defined in \eqref{eq:l_def_exact_k}, satifies the  upper bounds
\begin{eqnarray}
  |l(x, \nu)| &\leq& \bar{k}e^{\bar{k}} =: \bar{l}  \label{eq:regular_lbar_def}\,, \\
  |l_{\nu}(x, \nu)| &\leq& \bar{k}_{\nu}(1 + \bar{l})e^{\bar{k}} =: \bar{l}_{\nu} \,, \label{eq:regular_lbar_nu_def}
\end{eqnarray}
for every $(x, \nu) \in [0, 1] \times [-B_{\nu}, B_{\nu}]$ and where
\begin{equation}
  \bar{k} := B_{\beta}e^{B_{\beta}} \label{eq:bar_k_def}\,,
\end{equation}
and $\bar{k}_{\nu}$ is defined in \eqref{eq:k_nu_lyap}
\end{lemma}
\begin{proof}
  We use Lemma \ref{lemma_k_knu} and the successive approximation method to achieve \eqref{eq:regular_lbar_def} and the same for \eqref{eq:regular_lbar_nu_def} noticing that taking the derivative of \eqref{eq:l_def_exact_k} with respect to $\nu$ gives
  \begin{eqnarray}
    l_{\nu} (x, \nu) &=& k_{\nu}(x, \nu) + \int_{0}^{x} \bigg[k_{\nu}(x-y, \nu) l(y, \nu) \nonumber \\&&+ k(x-y, \nu) l_{\nu}(y, \nu) \bigg] dy \,, 
  \end{eqnarray}
  for all $(x, \nu) \in [0, 1] \times [-B_{\nu}, B_{\nu}]$.
\end{proof}

\begin{lemma}\label{lemma:lyap_loc_estimate_pbc}
  {\em [Loc. Lyapunov estimate pert. boundary conditions]}
  For any $\beta \in H$, for any $c \geq 1$ there exist positive functions $R_0(\breve B, c, B_{\nu}) = \mathcal{O}_{c \to \infty}(ce^{-c}), \epsilon^* = \mathcal{O}_{c \to \infty} (e^{-\frac{c}{2}})$ with a decreasing dependence on argument $\breve B$ such that for all $\epsilon \in (0, \epsilon^*)$ the Lyapunov function
  \begin{eqnarray}
    V(t) &=& V_1(t) + V_2(t) + V_3(t) \label{eq:V_pbc_def}\,,\\
    V_1(t) &=& \frac{c}{2} \int_{0}^{1} e^{cx} w^2(x, t) dx \label{eq:V1_pbc_def} \,, \\
    V_2(t) &=& \frac{c}{2} \int_{0}^{1} e^{cx} w_x^2(x, t) dx \label{eq:V2_pbc_def}\,, \\
    V_3(t) &=& \frac{1}{8}w^2(0, t) \,, \label{eq:V3_pbc_def}
  \end{eqnarray}
  satisfies
  \begin{equation}
    V(0) \leq R_0 \implies V(t) \leq V(0) e^{-\frac{c}{2}t}, \quad t \geq 0\,.
  \end{equation}
\end{lemma}

The result also holds for $c\in(0,1)$ but we don't spend the time proving it since it is for large $c$ (rapid decay) that the result is of interest. If it were interested to prove the result for $c\in(0,1)$, it would be easier to use the Lyapunov function
\begin{eqnarray}
 V(t) &=& \frac{1}{2} \int_{0}^{1} (1+x)w^2(x, t)dx \nonumber  \\  && + \frac{1}{2} \int_{0}^{1} (1+x)w_x^2(x, t)dx   + \frac{w^2(0, t)}{8}\,.
\end{eqnarray}
We would adapt the entire proof removing all terms in $c$ and get a bound in the form $V(t) \leq V(0) e^{-\frac{t}{2}}$. Noticing that $e^{-\frac{t}{2}} \leq e^{\frac{-ct}{2}}$ we would have the result valid for all $c > 0$.

With the definition used for $c \geq 1$, we have the equivalence with the norm $\Psi(w(t))$ introduced in \eqref{eq:new_psi_def}:
 \begin{eqnarray}
   V(t) &\leq& \frac{ ce^c}{2} \Psi(w(t))\,,  \\
   \Psi(w(t)) &\leq& 8 V(t)\, .
  \label{eq:equiv_new_psi_new_V}
 \end{eqnarray}
 We have to introduce the factor $c$ on $V_1, V_2$ to achieve an upper bound of the form $\dot{V} \leq -f(c) V(t)$ where $f(c) > 0$. Indeed one can take a look at \eqref{eq:w0_bound} and see that there is no $c$ factor on the terms in $w^2(0, t)$.

 \begin{proof}[Proof of Lemma \ref{lemma:lyap_loc_estimate_pbc}.]
We begin with the assumption---to be enforced shortly with a restriction on the initial condition---that  $ w^2(0, t) \leq B_{\nu}^2\,, \ t \geq 0$. Since
 \begin{equation}
  w^2(0, t) \leq 8V(t)\,, 
 \end{equation}
 if we chose $V(0) \leq \frac{B_{\nu}^2}{8}$ and prove that $V$ is decreasing, the assumption is validated for all $t \geq 0$.
We need to adapt the upper bounds \eqref{eq:delta_bound}-\eqref{eq:omega_x_bound} as well as introducing new ones.

With Lemmas \ref{lemma_l_bound} and \ref{lemma_k_knu}, when $|w(0, t)| \leq B_{\nu}$ we have 
\begin{eqnarray}
  \|l(\cdot, w(0, t))\|_{\infty} &\leq& \bar{l} \label{eq:bar_l_pbc} \,, \\
  \|l_{\nu}(\cdot, w(0, t))\|_{\infty} &\leq& \bar{l}_{\nu} \label{eq:bar_l_nu_pbc} \,, \\
  \|k_{\nu}(\cdot, w(0, t))\|_{\infty} &\leq& \bar{k}_{\nu} \,, \\
  \|k_{x \nu}(\cdot, w(0, t))\|_{\infty} &\leq& \bar{k}_{x \nu} \,, \\
  \|\tilde{k}(\cdot, w(0, t))\|_{\infty} &\leq& \epsilon \label{eq:tilde_k_pbc} \,, \\
  \|\tilde{k}_{\nu}(\cdot, w(0, t))\|_{\infty} &\leq& \epsilon \label{eq:tilde_k_nu_pbc} \,, \\
  \|\Omega(\cdot, t)\|_{\infty} &\leq& \bar{\Omega} \int_{0}^{1} |w(y, t)|dy \label{eq:omega_pbc_bound}\,, \\
  |\Omega_x(x, t)| &\leq& \bar{\Omega}_{x1} \int_{0}^{1} |w(y, t)|dy \nonumber \\ && + \bar{k}_{\nu} |w(x, t)| \label{eq:omegax_pbc_bound}\,, \\
  |\Gamma(t)| &\leq& \epsilon \bar{\Gamma} \int_{0}^{1} |w(y, t)| dy \label{eq:gamma_bound}\,, \\
  |\Gamma_t(t)| &\leq& \epsilon \bar{\Gamma}_1 |w_x(0, t)| \int_{0}^{1} |w(y, t)|dy \nonumber \\ && + \epsilon \bar{\Gamma} \int_{0}^{1} |w_x(y, t)|dy \label{eq:gammat_bound}\,, 
\end{eqnarray}
where $\bar{l}, \bar{l}_{\nu}$ are defined in Lemma \ref{lemma:inverse_kernel_bound} and $\bar{k}_{\nu}, \bar{k}_{x \nu}$ are defined in Lemma \ref{lemma_k_knu} and
$  \bar{\Omega} := \bar{k}_{\nu}(1 + \bar{l}) , 
  \bar{\Omega}_{x1} := \bar{k}_{\nu} \bar{l} + \bar{k}_{x \nu}(1 + \bar{l}) , 
  \bar{\Gamma} := 1 + \bar{l} , 
  \bar{\Gamma}_1 := 1 + \bar{l} + \bar{l}_{\nu} + \bar{\Omega}(1 +\bar{l})$. 
The proofs of \eqref{eq:omega_pbc_bound} and \eqref{eq:omegax_pbc_bound} are almost identical as the ones for \eqref{eq:omega_bound} and \eqref{eq:omega_x_bound} and so we focus on proving  \eqref{eq:gamma_bound} and \eqref{eq:gammat_bound}. 

\underline{For $\Gamma$:}
Using \eqref{eq:tilde_k_pbc}, \eqref{eq:bar_l_pbc} and the triangular inequality we have the desired result.  

\underline{For $\Gamma_t$:} Taking the derivative of \eqref{eq:gamma_def} we have that
\begin{eqnarray}
  - \Gamma_t(t) &=& w_t(0, t) \int_{0}^{1} \tilde{k}_{\nu}(1-y, w(0, t))  \times \nonumber \\ &&  \left [w(y, t) + \int_{0}^{y} l(y-s, w(0, t)) w(s, t) ds \right] dy \nonumber \\
                &&+ \int_{0}^{1} \tilde{k}(1-y, w(0, t)) \nonumber \\
                &&\times \bigg[w_x(y, t) - w_x(0, t) \Omega(y, t) \nonumber \\ && + w_t(0, t) \int_{0}^{y}l_{\nu}(y-s, w(0, t)) w(s, t) ds\bigg]dy \nonumber \\
                &&+ \int_{0}^{1} \tilde{k}(1-y, w(0, t)) \int_{0}^{y}l(y-s, w(0, t)) \nonumber \\ && \times  \left [w_x(s, t) - w_x(0, t) \Omega(s, t)\right ]ds dy\,.
\end{eqnarray}
Since $w_t(0, t) = w_x(0, t)$, using \eqref{eq:omegax_pbc_bound}, \eqref{eq:tilde_k_pbc}, \eqref{eq:tilde_k_nu_pbc}, \eqref{eq:bar_l_pbc}, \eqref{eq:bar_l_nu_pbc} and \eqref{eq:omega_pbc_bound} we arrive at \eqref{eq:gammat_bound}.
Also let's notice that Cauchy-Schwarz's inequality provide these very useful upper bounds
\begin{eqnarray}
  \int_{0}^{1} |w(x, t)| dx &\leq& \sqrt{\int_{0}^{1} w^2(x, t) dx} \label{eq:w_holder} \,, \\
  \int_{0}^{1} |w_x(x, t)| dx &\leq& \sqrt{\int_{0}^{1} w_x^2(x, t) dx} \label{eq:wx_holder}\,, \\
                              && \forall t \geq 0\,. \nonumber
\end{eqnarray}
Let's compute the derivative of $V$ for $t \geq 0$.

\underline{Estimate of $V_3$:}
Taking the derivative of \eqref{eq:V3_pbc_def}, using \eqref{eq:wt_def_exact_k}, $c \geq 1$ and Young's inequality we have that
\begin{equation}
  \dot{V}_3(t) \leq \frac{c w^2(0, t)}{4}  + \frac{c w_x^2(0, t)}{16}  \label{eq:V3_pbc_bound}\,.
\end{equation}

\underline{Estimate of $V_1$:} Taking the derivative of \eqref{eq:V1_pbc_def}, we have that 
\begin{eqnarray}
  \dot{V}_1 &=& I_{11} + I_{12} \,,\\
  I_{11}(t) &=& c\int_{0}^{1}e^{cx} w(x, t) w_x(x, t) dx\,,  \label{eq:I11_pbc_def} \\
  I_{12}(t) &=& -c w_x(0, t) \int_{0}^{1}e^{cx} w(x, t) \Omega(x, t) dx \label{eq:I12_pbc_def}\,.
\end{eqnarray}
We first work on  $I_{11}$. Integrating by parts,,
\begin{eqnarray}
  I_{11}(t) &=& \frac{ce^c}{2} \Gamma^2(t) - \frac{c}{2} w^2(0, t) - c V_1(t) \nonumber \\
            &\leq& \epsilon^2 \bar{\Gamma}^2 e^c  V_1(t)  - \frac{c}{2} w^2(0, t) - c V_1(t) \label{eq:I_11_pbc_bound}\,. 
\end{eqnarray}
The  bound \eqref{eq:I_11_pbc_bound} follows from \eqref{eq:gamma_bound} and \eqref{eq:w_holder}.
For $I_{12}$ we use \eqref{eq:omega_pbc_bound}, \eqref{eq:w_holder} as well as Young's inequality. We then have
\begin{eqnarray}
  I_{12}(t) \leq \frac{c}{8} w_x^2(0, t) + 4 \bar{\Omega}^2 e^c V_1(t) \|w(t)\|^2 \label{eq:I_12_pbc_bound}
\end{eqnarray}
Gathering \eqref{eq:I_11_pbc_bound} and \eqref{eq:I_12_pbc_bound} we have
\begin{eqnarray}
  \dot{V_1} (t) &\leq& \frac{c}{8} w_x^2(0, t) - \frac{c}{2} w^2(0, t) - V_1(t) \nonumber \\ && \times \left [c - \epsilon^2 \bar{\Gamma}^2 e^c - 4 \bar{\Omega}^2 e^c \|w(t)\|^2 \right]\,. 
  \label{eq:V_1_pbc_bound}
\end{eqnarray}

\underline{Estimate of $V_2$:} 
Taking the derivative with respect to $x$ of \eqref{eq:wt_def_exact_k} and using \eqref{eq:w_perturbed_bc} gives the following system satisfied by $w_x$
\begin{eqnarray}
  w_{xt}(x, t) &=& w_{xx}(x, t) - w_x(0, t) \Omega_x(x, t) \label{eq:wx_pbc}\,,  \\
  w_x(1, t) &=& \Gamma_t(t) + w_x(0, t) \Omega(1, t) \nonumber \label{eq:wx1_pbc} \,, \\
            && \forall (x, t) \in [0, 1] \times \mathbb{R}^+\,.
\end{eqnarray}
We can then take the derivative of $V_2$ (defined in \eqref{eq:V2_pbc_def}). We then have
\begin{eqnarray}
  \dot{V}_2 &=& I_{21} + I_{22} \,,  \\
  I_{21}(t) &=& c \int_{0}^{1} e^{cx} w_x(x, t) w_{xx}(x, t) dx \label{eq:I21_pbc_def} \,, \\
  I_{22}(t) &=& -c w_x(0, t) \int_{0}^{1} e^{cx} w_x(x, t) \Omega_x(x, t) dx \label{eq:I22_pbc_def} \,.
\end{eqnarray}
We first work $I_{21}$. Integrating by parts,
\begin{eqnarray}
  I_{21}(t) &=& \frac{ce^c}{2} w_x^2(1, t) - \frac{c}{2} w_x^2(0, t) -c V_2(t)
\end{eqnarray}
Using \eqref{eq:wx1_pbc}, \eqref{eq:gammat_bound}, \eqref{eq:omegax_pbc_bound}, \eqref{eq:wx_holder} and \eqref{eq:w_holder} we have that
\begin{eqnarray}
  |w_x(1, t)| &\leq& \epsilon \bar{\Gamma} \|w_x(t)\| + \epsilon \bar{\Gamma}_1 |w_x(0, t)|.\|w(t)\| \nonumber \\ && + \bar{\Omega} |w_x(0, t)|.\|w(t)\|, \quad t \geq 0\,.
\end{eqnarray}
We thus have for $t\geq 0$, 
\begin{align}
  I_{21}(t) \leq& e^c \bigg[ 4 \epsilon^2 \bar{\Gamma}^2 V_2(t)  + 4 \epsilon^2 w_x^2(0, t) V_1(t) \nonumber \\ & + 2 \bar{\Omega}^2 w_x^2(0, t) V_1(t) \bigg]  
  - \frac{c}{2}w_x^2(0, t) - c V_2(t)\,.  \label{eq:I21_bound_pbc}
\end{align} 
Moving to $I_{22}$, using \eqref{eq:omegax_pbc_bound}, Cauchy-Schwarz and Young's inequalities we have that
\begin{equation}
  I_{22}(t) \leq \frac{cw_x^2(0, t)}{4} + 4 e^c V_2(t) \|w(t)\|^2 (\bar{k}_{\nu}^2 + \bar{\Omega}_{x1}^2) \label{eq:I22_bound_pbc} \,.
\end{equation}
Gathering \eqref{eq:I21_bound_pbc}, \eqref{eq:I22_bound_pbc} we have that
\begin{align}
\hspace*{-1em}  
\dot{V}_2(t) \leq& - w_x^2(0, t) \left [\frac{c}{4} - 2 e^c \epsilon^2 V(t) (2 + \bar{\Omega}^2) \right ] \nonumber \\
               & - 
               \left [c - 4e^c(\epsilon^2 \bar{\Gamma}^2 + 2 V(t) (\bar{k}_{\nu}^2 + \bar{\Omega}^2))  \right ]V_2(t) 
  \label{eq:V2_dot_pbc_bound} \,. 
\end{align}
noticing that $V_1 \leq V, \, \|w(t)\|^2 \leq 2 V(t)$. 

\underline{Estimate of $V$:} 
Gathering \eqref{eq:V_1_pbc_bound}, \eqref{eq:V2_dot_pbc_bound} and \eqref{eq:V3_pbc_bound} we have 
\begin{align}
\hspace*{-1em} 
  \dot{V}(t) \leq& - w_x^2(0, t) \left [\frac{c}{16} - 2 e^c \epsilon^2 V(t) (2 + \bar{\Omega}^2)\right ] \label{eq:wx0_pbc_final_bound} \\ 
             & - V_1(t) \left [c - \epsilon^2 \bar{\Gamma}^2 e^c - 8 \bar{\Omega}^2 e^c V(t) \right ] \label{eq:V1_pbc_final_bound}\\
             & - V_2(t) \bigg [c - 4 e^c (\epsilon^2 \bar{\Gamma}^2 + 2 V(t) (\bar{k}_{\nu}^2 + \bar{\Omega}^2)) \bigg ] \label{eq:V2_pbc_final_bound} \\
             & - 2 c V_3(t) \label{eq:V3_pbc_final_bound} \,.
\end{align}
Considering \eqref{eq:V1_pbc_final_bound}, we introduce
\begin{eqnarray}
  R_1 &:=& \frac{ce^{-c}}{32 \bar{\Omega}^2} \quad > 0\,,  \\
  \epsilon_1 &:=& \frac{e^{-\frac{c}{2}}}{\bar{\Gamma}} \sqrt{\frac{c}{2} - 8 \bar{\Omega}^2 e^c R_1} \quad > 0\,.
\end{eqnarray}
That way if we choose $\epsilon \leq \epsilon_1, \, V(0) \leq R_1$, we would have \eqref{eq:V1_pbc_final_bound} upper bounded by $-\frac{c}{2} V_1(t)$. 
We then consider \eqref{eq:wx0_pbc_final_bound}, and introduce
\begin{equation}
  \epsilon_2 := {\rm min} \left \{ \epsilon_1, e^{-\frac{c}{2}} \sqrt{\frac{c}{32 R_1(2 + \bar{\Omega}^2)}} \right \} > 0\,,
\end{equation}
That way if we choose $\epsilon \leq \epsilon_1, \, V(0) \leq R_1$ we would have \eqref{eq:wx0_pbc_final_bound} negative and \eqref{eq:V1_pbc_final_bound} bounded by $-\frac{c}{2} V_1(t)$.
Finally we consider \eqref{eq:V2_pbc_final_bound} and introduce 
\begin{align}
  R_2 :=& {\rm min} \left \{R_1, \frac{ce^{-c}}{32 (\bar{k}_{\nu}^2 + \bar{\Omega}_{x1}^2)}\right \} \, > 0\,,  \\
  \epsilon^* :=& {\rm min} \bigg\{\epsilon_2, \frac{e^{\frac{-c}{2}}}{2 \bar{\Gamma}} 
  sqrt{\frac{c}{2} - 8 e^c R_2 (\bar{k}_{\nu}^2 + \bar{\Omega}_{x1}^2)} \bigg\}\,  > 0 \,. \label{eq:epsilon_start_pbc}
\end{align}
That way if we choose $\epsilon \leq \epsilon^*$, $V(0) \leq R_2$ we would have \eqref{eq:wx0_pbc_final_bound} $\leq 0$, \eqref{eq:V1_pbc_final_bound} $\leq -\frac{c}{2} V_1(t)$ and \eqref{eq:V2_pbc_final_bound} $\leq -\frac{c}{2} V_2(t)$. We make the final choice
$  \epsilon \leq \epsilon^* $ and 
\begin{eqnarray}
  V(0) &\leq& {\rm min} \left \{R_2,  \frac{B_{\nu}^2}{8} \right \} =: R_0 > 0\,, \label{eq:R0_pbc_def}\,, 
\end{eqnarray}
where the last condition in \eqref{eq:R0_pbc_def} is made to ensure that $w^2(0, t) \leq B_{\nu}^2$. With this choice we have that
\begin{equation}
  \dot{V}(t) \leq -\frac{c}{2} V(t), \ \forall t \geq 0\,.
\end{equation}
With the comparison principle we complete the proof.
 \end{proof}

\begin{lemma}\label{lemma:loc_exp_stability_target_pbc}
  For all $c \geq 1$ there exist functions $\epsilon^*(\breve B, c) = \mathcal{O}_{c \to \infty} (e^{-\frac{c}{2}}), \Psi_0(B_{\nu}, \breve B, c) = \mathcal{O}_{c \to \infty}(e^{-2c})$ with a decreasing dependence in the argument $\breve B$ such that, for any $\beta \in H$ and any NO approximating $k, k_{\nu}$ with any $\epsilon \in (0, \epsilon^*)$, where $\epsilon^*$ is defined in \eqref{eq:epsilon_start_pbc}, if the perturbed target system \eqref{eq:w_perturbed_bc}, \eqref{eq:wt_pbc_def} is initialized with $w_0 := w(., 0)$ such that $\Psi(w_0) \leq \Psi_0$, then 
  \begin{equation}
    \Psi(t) \leq 4 ce^c \Psi(0) e^{-\frac{c}{2}t}, \quad t \geq 0\,,
  \end{equation}
  where 
  \begin{equation}
   \Psi(w(t)) = \|w(t)\|^2 + \|w_x(t)\|^2 + w^2(0, t)\,.
  \end{equation}
\end{lemma}

\begin{proof}
  We just have to notice that the Lyapunov function defined in \eqref{eq:V_pbc_def} satisfies
$      V(t)  \leq \frac{ce^c}{2} \Psi(w(t))$ and $  \Psi(w(t)) \leq  8 V(t)$.
  We then use Lemma \ref{lemma:loc_exp_stability_target_pbc}, keeping $\epsilon^*$, and introducing $\Psi_0 := \frac{2e^{-c}}{c} R_0$, where $R_0$ is defined in Lemma. \ref{lemma:loc_exp_stability_target_pbc} 
\end{proof}

\subsection{Local Stability in Original Plant Variable} \label{subsection:equiv_perturbed_base}
We  come back to  system \eqref{eq:ut_def}, \eqref{eq:u1_def}, working with the norm
\begin{equation}
  \Omega(u(t)) := \|u(t)\|^2 + \|u_x(t)\|^2 + u^2(0,t)\,
  \label{eq:omega_def_pbc}
\end{equation}
and state the norm equivalence  between the perturbed target system \eqref{eq:wt_pbc_def}, \eqref{eq:w1_def_exact_k} and the original system \eqref{eq:ut_def}, \eqref{eq:u1_def}. 

\begin{lemma}\label{lemma:equiv_norm_target_base_pbc}
  {\em [Equiv. norm perturbed-target system]}
  There exist  $\rho, \delta>0$ such that, $\forall t \geq 0$, if $|u(0, t)| \leq B_{\nu}$ then,
  \begin{eqnarray}
    \Psi(w(t)) &\leq& \delta \,\Omega(u(t))\,,\label{eq:psi_bound_pbc} \\
    \Omega(u(t)) &\leq& \rho\, \Psi(w(t))\,.\label{eq:omega_norm_bound_pbc}
  \end{eqnarray}
\end{lemma}

\begin{proof}
  Noticing that $w(0, t) = u(0, t)$, we  adapt the proof of Lemma \ref{lemma:equiv_norm_target_base} and choose 
$      \delta := (6 + 4 \bar{k}^2 + 4\bar{k}_x^2) $ and $ \rho := (6 + 4 \bar{l}^2 + 4 \bar{l}^2)$, 
  where $\bar{l}$ is defined in Lemma \ref{lemma_l_bound} and
$      \bar{k} := B_{\beta}e^{B_{\beta}}, 
    \bar{k}_x := B_{\beta_x} (1 + \bar{k}) + B_{\beta} \bar{k}, 
    \bar{l}_x := \bar{k}_x + \bar{l} \bar{k} + \bar{k}_x \bar{l}$.
\end{proof}

With all the lemmas, we  now prove Theorem \ref{theorem:loc_stability_base_exact_k}.

\begin{proof}[Theorem \ref{theorem:loc_stability_base_exact_k}]
  Adapting the proof of Thm. \ref{theorem:loc_stability_base},  choose 
  \begin{eqnarray}
    \Omega_0 &:=& {\rm min}\left \{B_{\nu}^2, \frac{\Psi_0}{\delta} \right \}\,,  \\
    M &:=& \rho \delta c e^c \,, 
  \end{eqnarray}
  where $\rho, \delta$ are defined in Lemma \ref{lemma:equiv_norm_target_base_pbc} and $\Psi_0$ defined in Lemma \ref{lemma:loc_exp_stability_target_pbc}. We complete the proof 
  with 
  $\epsilon^*$ as defined in Lemma \ref{lemma:loc_exp_stability_target_pbc}.
\end{proof}

\section{Simulations}\label{sec:simulations}

\definecolor{lighter-green}{rgb}{0, 0.6, 0.00392156862} 
\begin{table}[t]
\label{tab:nopspeedups}
\centering
\resizebox{\columnwidth}{!}{%
\begin{tabular}{lccc}
\hline
\textbf{\begin{tabular}[c]{@{}l@{}}Spatial Step \\ Size (dx)\end{tabular}} & \multicolumn{1}{l}{\textbf{\begin{tabular}[c]{@{}l@{}}Analytical \\ Kernel \\ Calculation \\ Time(s) $\downarrow$ \end{tabular}}} & \multicolumn{1}{l}{\textbf{\begin{tabular}[c]{@{}l@{}}Neural Operator\\ Kernel\\ Calculation \\ Time(s) $\downarrow$ \end{tabular}}} & \multicolumn{1}{l}{\textbf{Speedup} $\uparrow$} \\ \hline
$0.01$                                                                     & $0.043$                                                                                                             & $0.028$                                                                                                                & {\color{lighter-green} $1.53$x}       \\
$0.001$                                                                    & $2.6$                                                                                                               & $0.029$                                                                                                                & {\color{lighter-green} $90$x}         \\
$0.0005$                                                                    & $10$                                                                                                                & $0.030$                                                                                                                & {\color{lighter-green} $333$x}        \\
$0.0001$                                                                   & $245$                                                                                                               & $0.057$                                                                                                                & {\color{lighter-green} $4298$x}       \\ \hline
\end{tabular}%
} 
\caption{Neural operator speedups over the analytical kernel calculation with respect to the increase in discretization points (decrease in step size). }
\end{table}

\begin{figure*}[t]
    \centering\includegraphics[width=.95\textwidth]{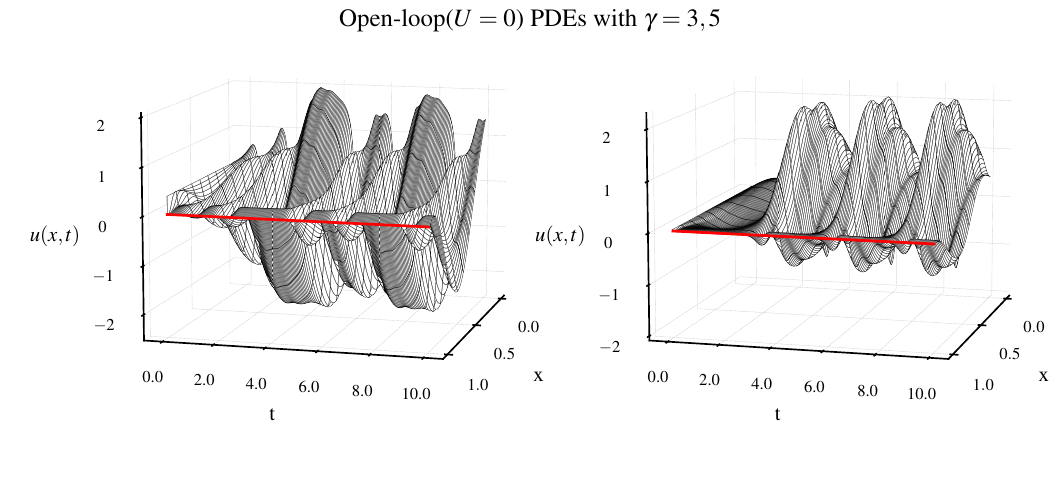}
    \caption{Open-loop (U=0) simulation of the PDE with the modified Chebyshev polynomial functions $\beta(x, u(0, t)) = 5 \cos((\gamma+u(0, t))\cos^{-1}(x))$ with initial conditions $u(x, 0) = 0.38, 0.04$ and parameters $\gamma=3, 5$ for the let and right images respectively.}
    \label{fig:gs-openloop}
\end{figure*}

\begin{figure*}[t]
    \centering
    \includegraphics[width=.95\textwidth]{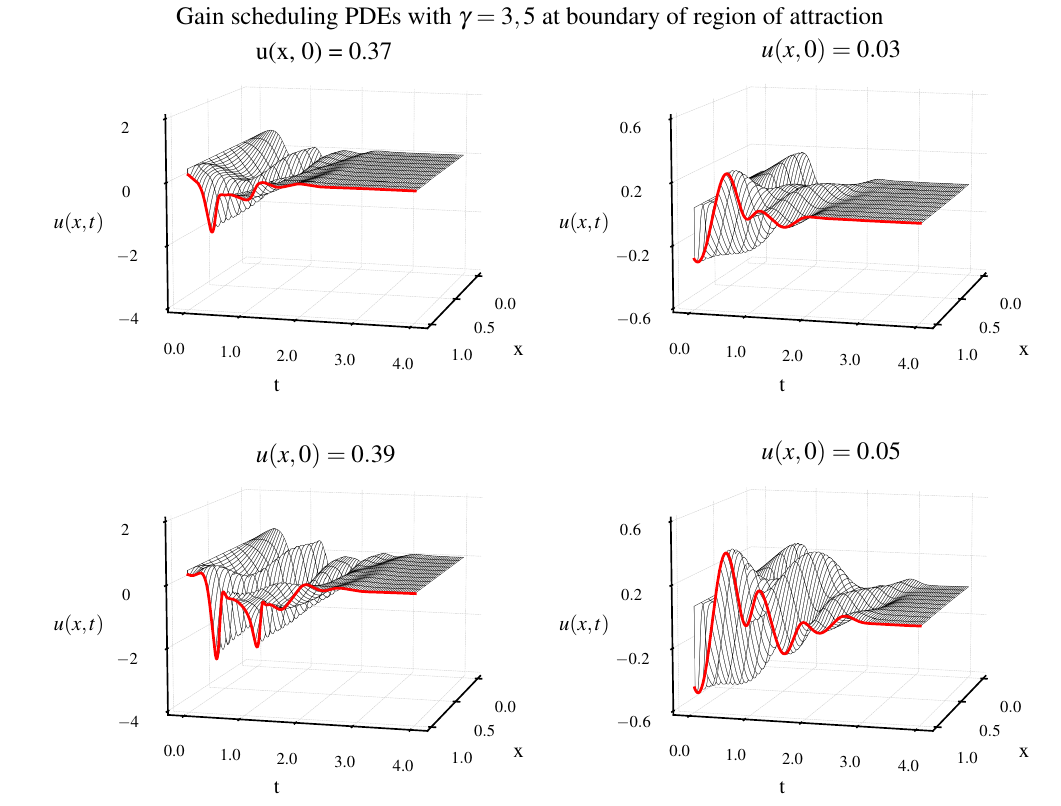}
    \caption{Analytical solution with gain scheduling for the modified Chebyshev polynomial functions $\beta(x, \mu) = 5 \cos ((\gamma+\mu)  \cos^{-1}(x))$ with parameters $\gamma=3, 5$ for the left and right images respectively. The top row shows initial conditions 0.37 (left) and 0.03 (right) respectively and the bottom row shows increased initial conditions of 0.39 (left) and 0.05(right). Naturally, the PDE becomes harder to stabilize and for larger initial conditions, gain scheduling fails to control the PDE.}
    \label{fig:gainScheduleROA}
\end{figure*}

\begin{figure*}[t]
    \centering
    \includegraphics[width=.95\textwidth]{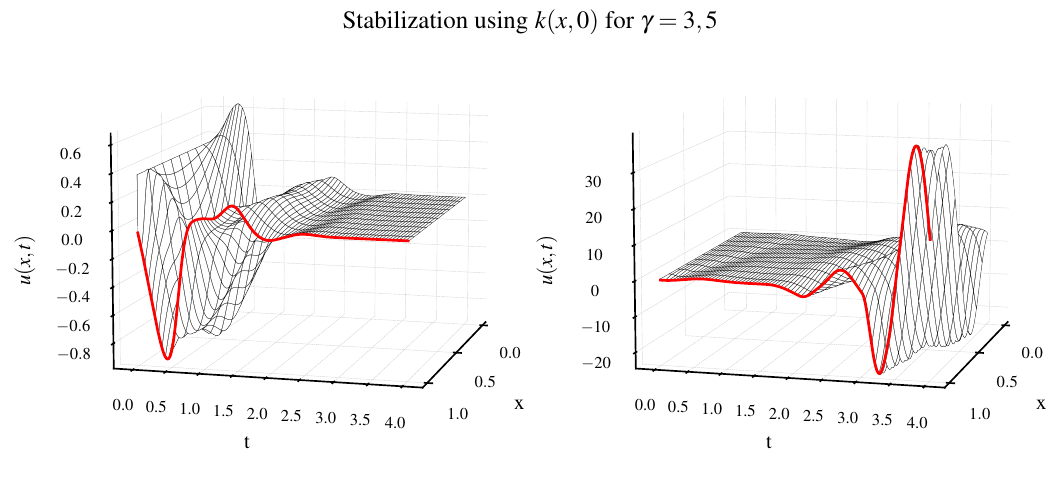}
    \caption{PDE stabilization using a controller based on a linearization at the origin, with no gain scheduling,  $U(t) = \int_0^x k(1-y, 0)u(y, t) dy$, for the PDE's corresponding to Figure \ref{fig:gs-openloop}}
    \label{fig:gainScheduleLinear}
\end{figure*}

\begin{figure*}
    \centering
    \includegraphics[width=.95\textwidth]{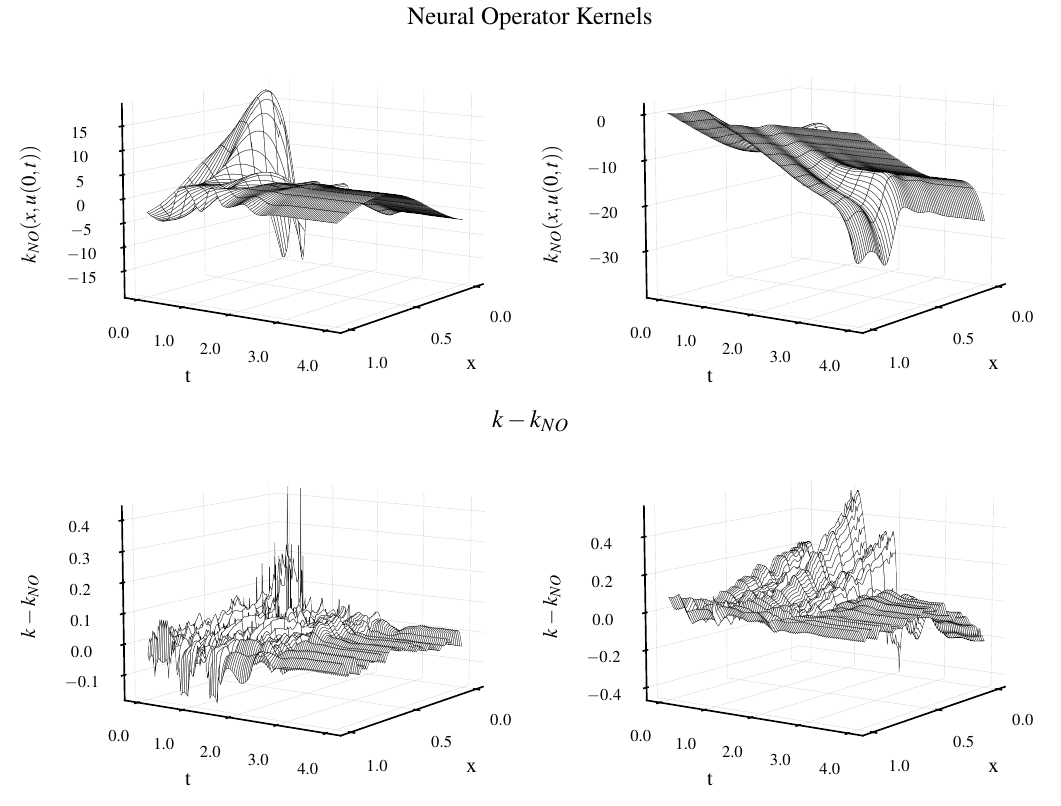}
    \caption{Neural operator kernels when controlling the PDE in Figure \ref{fig:gs-openloop}(top row), and the difference in kernel error between the resulting gain scheduling kernels (bottom row).}
    \label{fig:gainScheduleKernel}
\end{figure*}

\begin{figure*}[t]
    \centering
    \includegraphics[width=.95\textwidth]{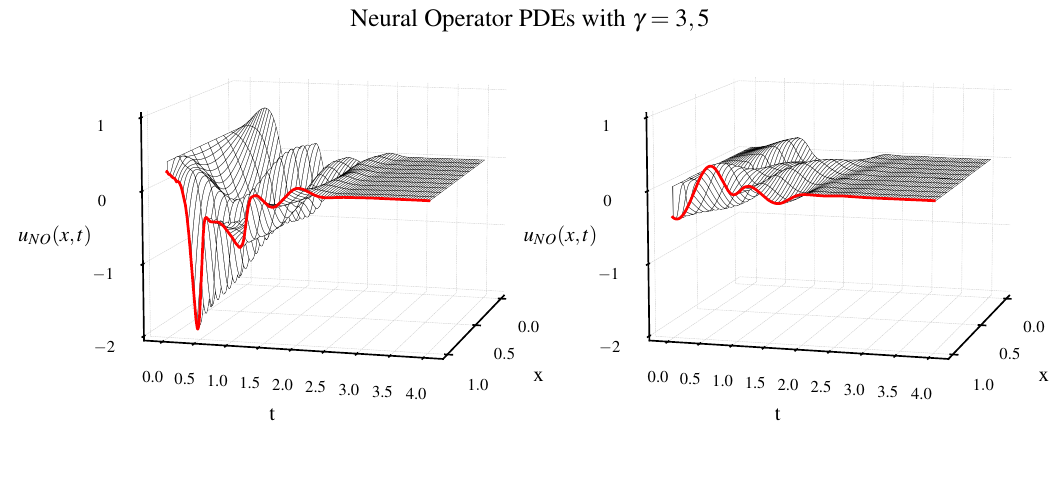}
    \caption{Stabilization of the plants in Figure \ref{fig:gs-openloop} with the neural operator approximated kernels in the gain scheduling feedback law.}
    \label{fig:gainScheduleStabilize}
\end{figure*}

We consider the simulation of Eq \eqref{eq:ut_def}, \eqref{eq:u1_def} where the recirculation function, $\beta$, is defined as $\beta(x, u(0, t)) = 5\cos((\gamma+u(0, t))\cos^{-1}(x))$ representing a Chebyshev polynomial of the first kind. 
Note that the $\gamma+u(0, t)$ term controls the resulting shape of the polynomial where $\gamma+u(0, t)=1$ is a line, $\gamma+u(0, t)=2$ is a parabola and so on. To simulate the PDE, we use a traditional first-order finite difference upwind scheme with temporal step $dt=1e-4$ and spatial step $dx=1e-2$. In Figure \ref{fig:gs-openloop}, we present the resulting PDE with open-loop control ($U=0$) and note that the dynamics result in a limit cycle. Then, in Fig. \ref{fig:gainScheduleROA}, we demonstrate the region of attraction of the proposed gain scheduling controller when implemented with the analytical kernel. We emphasize that, as theoretically shown, gain scheduling is locally stabilizing, but as shown in Fig. \ref{fig:gainScheduleLinear}, it provides better control in contrast to the failure of a linear backstepping controller designed for a linearization at the origin. In particular, Fig. \ref{fig:gainScheduleLinear} explicitly demonstrates the poor performance of the purely linear controller $U(t) = \int_0^x k(1-y, 0)u(y, t)dy$ as it fails to control the system on the right while performing poorly when compared with the gain-scheduled controlled for the system on the left. Thus, Figures \ref{fig:gainScheduleROA} and \ref{fig:gainScheduleLinear} emphasize the known results of \cite{SiranosianAntranikA2011GSBC} and \cite{doi:10.1137/20M1341167}. However, the authors reiterate that these simulations are extremely expensive to simulate due to the recalculation of the analytical kernel at every timestep. Thus, there is significant need for computational speedup to apply the methodology for any real-world application.

We now discuss the training process for the neural operator approximation of the kernel function. To effectively train the kernel operator, one must delicately select a diverse dataset in $\beta(x, \nu), x\in [0, 1], \nu \in \mathbb{R}$ in order to have significant coverage for the possible $u(0, t)$ boundary values as the system evolves in time. To effectively build a diverse dataset, we jointly sample $100$ $\gamma$ values from $\gamma \sim \text{Uniform}(3, 8)$ and for each one of those $\gamma$ values, we sample $200$ $\nu$ values from $\nu \sim \text{Uniform}(-5, 5)$ to create a total dataset size of $20000$ points. We emphasize that this large dataset is needed for sufficient coverage across $u(0, t)$ and if one plans to use neural operator approximations with larger initial condition, then one needs to expand their range of $\nu$ and most likely needs to increase the number of $\nu$ samples per $\gamma$ value.

Additionally, as theoretically discussed in both Section \ref{sec:gain-only} and \ref{sec:approximatek}, the operator mapping requires the approximation of the derivatives with respect to both $x$, $\nu$ and the mixed double derivative with respect to both arguments. In practice, we found that the resulting neural operator kernel is best trained to just approximate the mapping from $\beta(x, \nu) \mapsto k$ without the derivatives. This is for two reasons. First, from a training, implementation, and speedup perspective, the simpler the mapping for computational calculation, the better the performance. For example, the gain-only approach in Section \ref{sec:gain-only} would double the network output size to include calculation of the derivative $\frac{\partial \mathcal{K}}{\partial \nu}$ and the full operator approach in Section \ref{sec:approximatek} would quadruple the output layer size to include all the derivatives. Second, from a performance perspective, we found that including the derivatives significantly worsened the computational performance of the $\mathcal{K}$ approximation. This is due to the large difference in derivative magnitudes (values can range from $10^3$ to $10^{-5}$ for various combinations in the given dataset) dominating the loss function causing the optimizer to favor accurately approximating the derivative over the operator solution $\mathcal{K}$ even when standard normalization techniques are applied. Furthermore, the authors further considered an alternative physics informed neural network (PINN) approach \cite{wang2021physicsinformedNeuralOperators}, \cite{RAISSI2019686}, \cite{li2023physicsinformed} where the mapping learned is still $\beta(x, \nu) \mapsto k$, but the loss function includes a weighted penalty term that penalizes when the $\mathcal{K}$ derivatives are different from the analytical solution using a $5$ point finite difference stenciling \cite{abramowitz+stegun}. After tuning the penalty weighting, this approach leads to better results than learning the entire operator, but the increase in training time (almost $2$x) is expensive and the error on just $\mathcal{K}$ is still larger then learning the mapping without any derivative penalty terms.

The training of the neural operator uses a modification of the DeepXDE package \cite{lu2021deepxde} and takes approximately $200$ seconds to train on a Nvidia RTX $3090$. The neural operator trained has $282,625$ parameters with the branch net as a traditional $4$-layer multi-layer perception (MLP) network and the trunk net containing a $2$-layer MLP. For ease of future research, all of the hyperparameter details are readily available on \href{https://github.com/lukebhan/NeuralOperatorApproximatedGainScheduling}{Github} in an accessible Jupyter notebook. The kernel $L_2$ training error was $1.96\times 10^{-3}$ and the $L_2$ testing error was $2.04\times 10^{-3}$. Figure \ref{fig:gainScheduleKernel} shows both the analytical kernel (top row), learned kernel (middle row), and the error (bottom row) with a maximum of approximately $10\%$ between the analytical and learned kernel when applied to the plant in \ref{fig:gs-openloop}. However, this error becomes very small when taken in the broader control feedback loop. For example, in Figure \ref{fig:gainScheduleStabilize}, one can see that the closed-loop feedback control with the neural operator approximated kernel effectively stabilizes both systems (top row) with minimal error when compared to stabilization with the analytical kernel (bottom row). Lastly, we conclude the discussion of neural operator approximated simulations by presenting the neural operator approximation performance speedups for various step sizes in Table \ref{tab:nopspeedups}. As expected, the neural operator performance gained scales as the number of discretization points increases (spatial step reduces) and for steps sizes of $dx=10^{-3}, 10^{-4}$, speedups are on magnitudes of $10^2, 10^3$ resulting in a reduction of \emph{4 mins to 0.5 seconds} for kernel calculation. We emphasize that this is for a single kernel calculation which needs to be completed \emph{at every single timestep} and thus every speedup enhancement gets compounded as the time horizon increases.

\section{Conclusion}
In this paper, we capitalize on the initial framework in \cite{bhan_neural_2023} for gain-scheduling of nonlinear hyperbolic PDEs. We introduce the gain-schedulable kernel operator and show that is can be accurately approximated, both theoretically and numerically, by a DeepONet. Then, we present an $H_1$-norm Lyapunov analysis to prove local stability of the resulting gain-scheduled controller. We also present a similar result for the "gain-only" approach where only the boundary component of the gain kernel is approximated. We conclude by showcasing simulations of a stabilizing gain-scheduled control law for two challenging hyperbolic PDE problems with nonlinear Chebyshev recirculating coefficients. The resulting simulations demonstrate the operator approximated gain functions achieve numerical speedups of the order of magnitude of $10^3$ \emph{at every single timestep in the calculation} paving the way for real-time implementation of gain scheduling feedback laws.

\section*{References} 
\vspace*{-1em}
\bibliographystyle{IEEEtranS}
\bibliography{references}

\begin{thebibliography}{10}
\providecommand{\url}[1]{#1}
\csname url@samestyle\endcsname
\providecommand{\newblock}{\relax}
\providecommand{\bibinfo}[2]{#2}
\providecommand{\BIBentrySTDinterwordspacing}{\spaceskip=0pt\relax}
\providecommand{\BIBentryALTinterwordstretchfactor}{4}
\providecommand{\BIBentryALTinterwordspacing}{\spaceskip=\fontdimen2\font plus
\BIBentryALTinterwordstretchfactor\fontdimen3\font minus \fontdimen4\font\relax}
\providecommand{\BIBforeignlanguage}[2]{{%
\expandafter\ifx\csname l@#1\endcsname\relax
\typeout{** WARNING: IEEEtranS.bst: No hyphenation pattern has been}%
\typeout{** loaded for the language `#1'. Using the pattern for}%
\typeout{** the default language instead.}%
\else
\language=\csname l@#1\endcsname
\fi
#2}}
\providecommand{\BIBdecl}{\relax}
\BIBdecl

\bibitem{abramowitz+stegun}
M.~Abramowitz and I.~A. Stegun, \emph{Handbook of Mathematical Functions with Formulas, Graphs, and Mathematical Tables}, 9th~ed.\hskip 1em plus 0.5em minus 0.4em\relax New York: Dover, 1964.

\bibitem{203622}
A.~Banach and W.~Baumann, ``Gain-scheduled control of nonlinear partial differential equations,'' in \emph{29th IEEE Conference on Decision and Control}, 1990, pp. 387--392 vol.2.

\bibitem{bekiaris2019nonlinear}
N.~Bekiaris-Liberis and R.~Vazquez, ``Nonlinear bilateral output-feedback control for a class of viscous {H}amilton--{J}acobi {PDE}s,'' \emph{Automatica}, vol. 101, pp. 223--231, 2019.

\bibitem{bhan_neural_2023}
L.~Bhan, Y.~Shi, and M.~Krstic, ``Neural operators for bypassing gain and control computations in {PDE} backstepping,'' \emph{IEEE Transactions on Automatic Control}, pp. 1--16, 2023.

\bibitem{392253}
T.~Chen and H.~Chen, ``Universal approximation to nonlinear operators by neural networks with arbitrary activation functions and its application to dynamical systems,'' \emph{IEEE Transactions on Neural Networks}, vol.~6, no.~4, pp. 911--917, 1995.

\bibitem{Coron2013Local}
J.~Coron, R.~Vazquez, M.~Krstic, and G.~Bastin, ``Local exponential {$H^2$} stabilization of a $2\times2$ quasilinear hyperbolic system using backstepping,'' \emph{SIAM Journal on Control and Optimization}, vol.~51, no.~3, pp. 2005--2035, 2013.

\bibitem{lu2021advectionDeepONet}
B.~Deng, Y.~Shin, L.~Lu, Z.~Zhang, and G.~E. Karniadakis, ``Approximation rates of deeponets for learning operators arising from advection–diffusion equations,'' \emph{Neural Networks}, vol. 153, pp. 411--426, 2022.

\bibitem{AIHPC_2000__17_5_583_0}
E.~Fern\'andez-Cara and E.~Zuazua, ``\BIBforeignlanguage{en}{Null and approximate controllability for weakly blowing up semilinear heat equations},'' \emph{\BIBforeignlanguage{en}{Annales de l'I.H.P. Analyse non lin\'eaire}}, vol.~17, no.~5, pp. 583--616, 2000.

\bibitem{doi:10.1137/20M1341167}
I.~Karafyllis, N.~Espitia, and M.~Krstic, ``Event-triggered gain scheduling of reaction-diffusion {PDE}s,'' \emph{SIAM Journal on Control and Optimization}, vol.~59, no.~3, pp. 2047--2067, 2021.

\bibitem{doi:10.1137/19M1252235}
I.~Karafyllis and M.~Krstic, ``Global stabilization of a class of nonlinear reaction-diffusion partial differential equations by boundary feedback,'' \emph{SIAM Journal on Control and Optimization}, vol.~57, no.~6, pp. 3723--3748, 2019.

\bibitem{9744516}
------, ``Spill-free transfer and stabilization of viscous liquid,'' \emph{IEEE Transactions on Automatic Control}, vol.~67, no.~9, pp. 4585--4597, 2022.

\bibitem{krstic2008Backstepping}
M.~Krstic and A.~Smyshlyaev, ``Backstepping boundary control for first-order hyperbolic {PDE}s and application to systems with actuator and sensor delays,'' \emph{Systems $\&$ Control Letters}, vol.~57, no.~9, pp. 750--758, 2008.

\bibitem{krstic2008boundary}
------, \emph{Boundary Control of {PDE}s: A Course on Backstepping Designs}.\hskip 1em plus 0.5em minus 0.4em\relax SIAM, 2008.

\bibitem{krstic2023neural}
M.~Krstic, L.~Bhan, and Y.~Shi, ``Neural operators of backstepping controller and observer gain functions for reaction-diffusion {PDE}s,'' 2023, arXiv:2303.10506.

\bibitem{krstic2008nonlinear}
M.~Krstic, L.~Magnis, and R.~Vazquez, ``Nonlinear stabilization of shock-like unstable equilibria in the viscous {B}urgers {PDE},'' \emph{IEEE Transactions on Automatic Control}, vol.~53, no.~7, pp. 1678--1683, 2008.

\bibitem{lanthaler2023nonlocal}
S.~Lanthaler, Z.~Li, and A.~M. Stuart, ``The nonlocal neural operator: Universal approximation,'' 2023, arXiv:2304.13221.

\bibitem{li2023physicsinformed}
Z.~Li, H.~Zheng, N.~Kovachki, D.~Jin, H.~Chen, B.~Liu, K.~Azizzadenesheli, and A.~Anandkumar, ``Physics-informed neural operator for learning partial differential equations,'' 2023, arXiv:2111.03794.

\bibitem{Lu2021DeepONet}
L.~Lu, P.~Jin, G.~Pang, Z.~Zhang, and G.~E. Karniadakis, ``Learning nonlinear operators via deeponet based on the universal approximation theorem of operators,'' \emph{Nature Machine Intelligence}, vol.~3, no.~3, pp. 218--229, 2021.

\bibitem{lu2021deepxde}
L.~Lu, X.~Meng, Z.~Mao, and G.~E. Karniadakis, ``{DeepXDE}: A deep learning library for solving differential equations,'' \emph{SIAM Review}, vol.~63, no.~1, pp. 208--228, 2021.

\bibitem{qi2023neural}
J.~Qi, J.~Zhang, and M.~Krstic, ``Neural operators for delay-compensating control of hyperbolic {PIDE}s,'' 2023, arXiv:2307.11436.

\bibitem{RAISSI2019686}
M.~Raissi, P.~Perdikaris, and G.~Karniadakis, ``Physics-informed neural networks: A deep learning framework for solving forward and inverse problems involving nonlinear partial differential equations,'' \emph{Journal of Computational Physics}, vol. 378, pp. 686--707, 2019.

\bibitem{RUGH20001401}
W.~J. Rugh and J.~S. Shamma, ``Research on gain scheduling,'' \emph{Automatica}, vol.~36, no.~10, pp. 1401--1425, 2000.

\bibitem{103361}
W.~Rugh, ``Analytical framework for gain scheduling,'' \emph{IEEE Control Systems Magazine}, vol.~11, no.~1, pp. 79--84, 1991.

\bibitem{58498}
J.~Shamma and M.~Athans, ``Analysis of gain scheduled control for nonlinear plants,'' \emph{IEEE Transactions on Automatic Control}, vol.~35, no.~8, pp. 898--907, 1990.

\bibitem{SiranosianAntranikA2011GSBC}
A.~A. Siranosian, M.~Krstic, A.~Smyshlyaev, and M.~Bement, ``\BIBforeignlanguage{eng}{Gain scheduling-inspired boundary control for nonlinear partial differential equations},'' \emph{\BIBforeignlanguage{eng}{Journal of dynamic systems, measurement, and control}}, vol. 133, no.~5, p.~12, 2011.

\bibitem{VAZQUEZ20082778}
R.~Vazquez and M.~Krstic, ``Control of 1-{D} parabolic {PDE}s with {V}olterra nonlinearities, {P}art {I}: Design,'' \emph{Automatica}, vol.~44, no.~11, pp. 2778--2790, 2008.

\bibitem{VAZQUEZ20082791}
------, ``Control of 1-{D} parabolic {PDE}s with {V}olterra nonlinearities, {P}art {II}: Analysis,'' \emph{Automatica}, vol.~44, no.~11, pp. 2791--2803, 2008.

\bibitem{wang2023deep}
S.~Wang, M.~Diagne, and M.~Krstić, ``Deep learning of delay-compensated backstepping for reaction-diffusion {PDE}s,'' 2023, arXiv:2308.10501.

\bibitem{wang2021physicsinformedNeuralOperators}
S.~Wang, H.~Wang, and P.~Perdikaris, ``Learning the solution operator of parametric partial differential equations with physics-informed deeponets,'' \emph{Science Advances}, vol.~7, no.~40, p. eabi8605, 2021.

\end{thebibliography}


\begin{IEEEbiography}[{\includegraphics[width=1in,height=1.25in,clip,keepaspectratio]{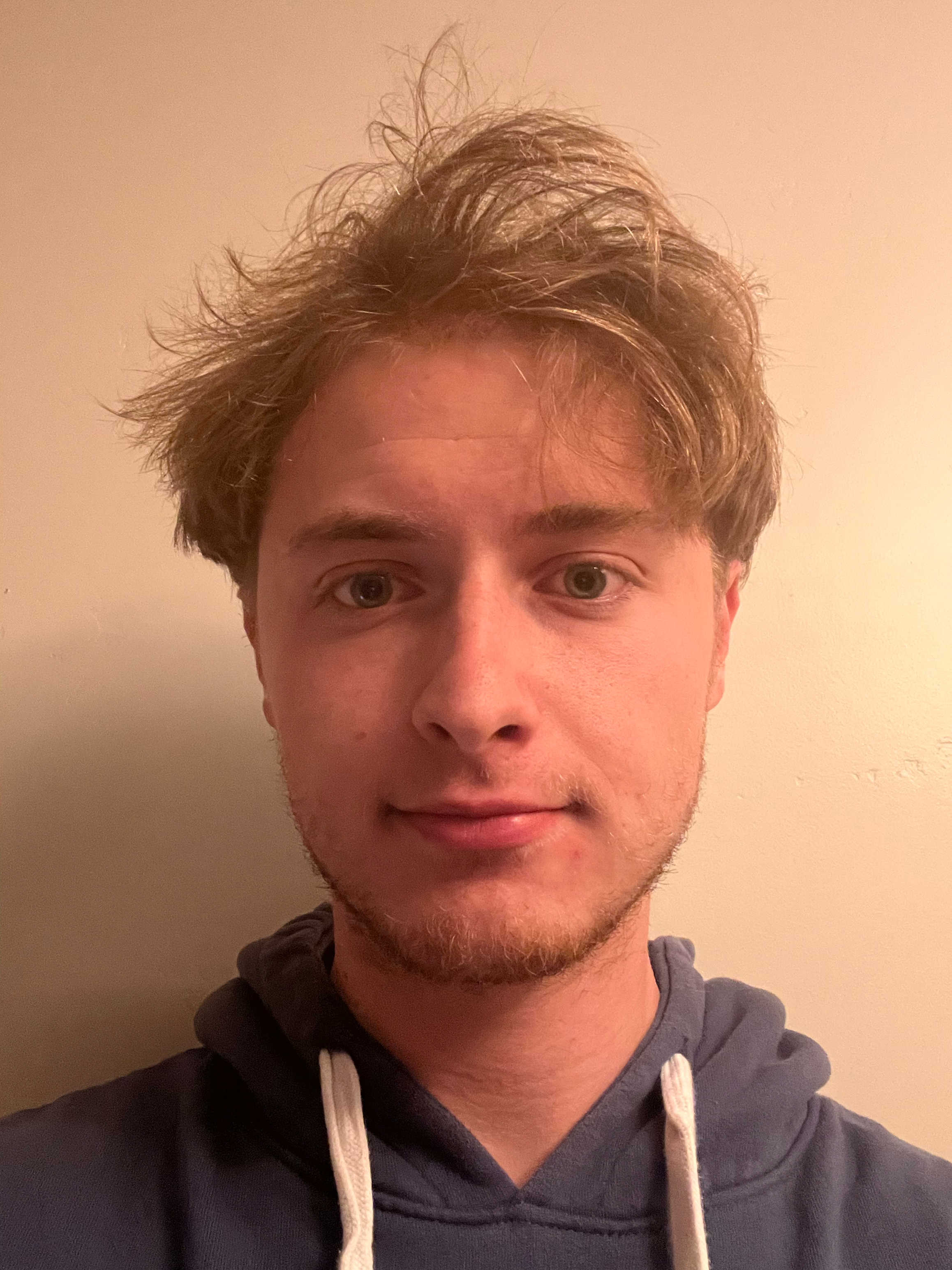}}]{Maxence Lamarque} is a graduate student at \'{E}cole des Mines
de Paris, France, due to graduate with a Master’s degree in general engineering in 2025. He was a visiting student at University of California, San Diego, in Fall 2023. 

\end{IEEEbiography}

\begin{IEEEbiography}[{\includegraphics[width=1in,height=1.25in,clip,keepaspectratio]{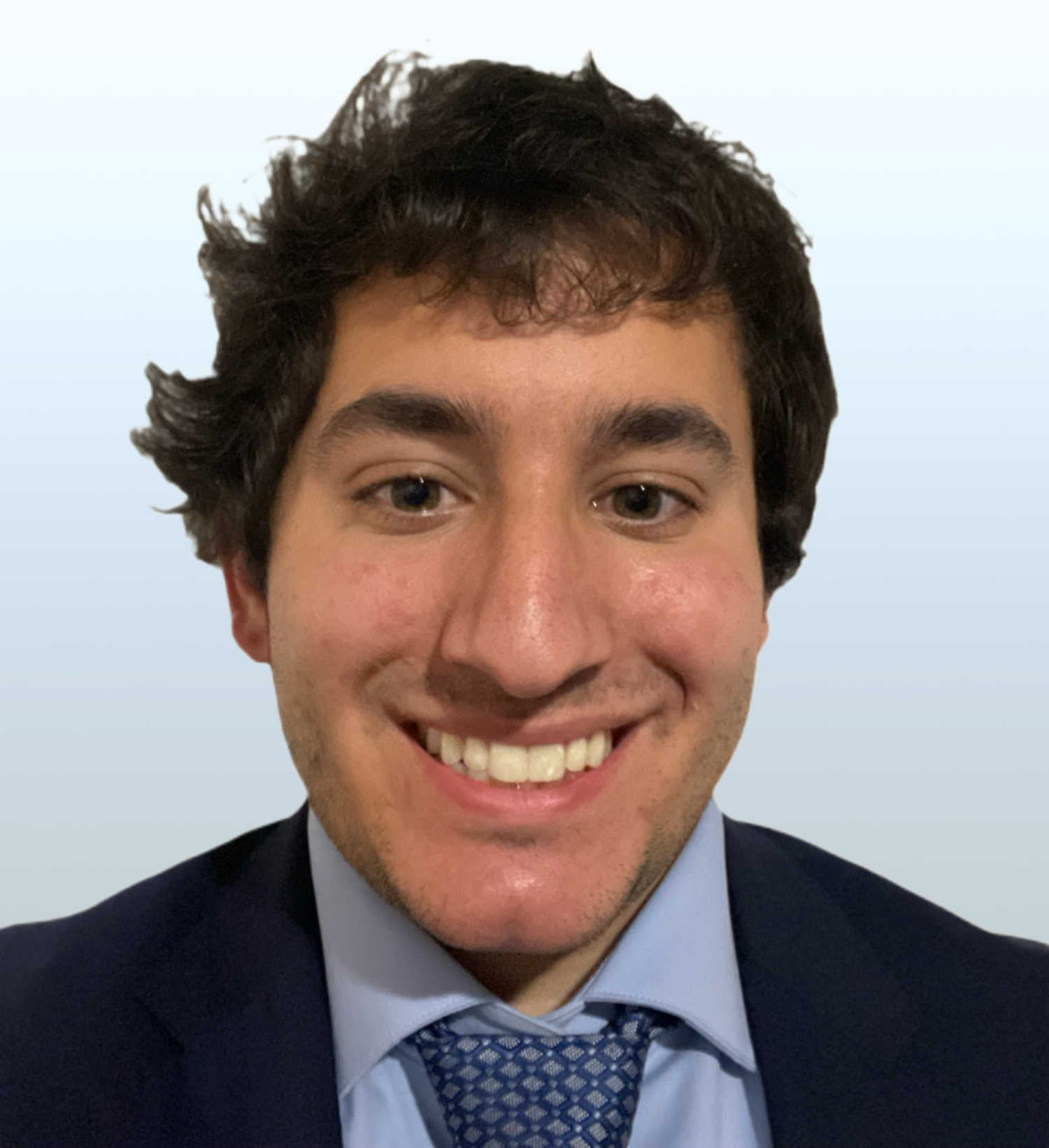}}]{Luke Bhan} received his B.S. and M.S. degrees in Computer Science and Physics from Vanderbilt University in 2022. He is currently pursuing his Ph.D. degree in Electrical and Computer Engineering
at the University of California, San Diego. His
research interests include neural operators, learning-based control, and control of partial differential equations. 
 \end{IEEEbiography}

\begin{IEEEbiography}
[{\includegraphics[width=1in,height=1.25in,clip,keepaspectratio]{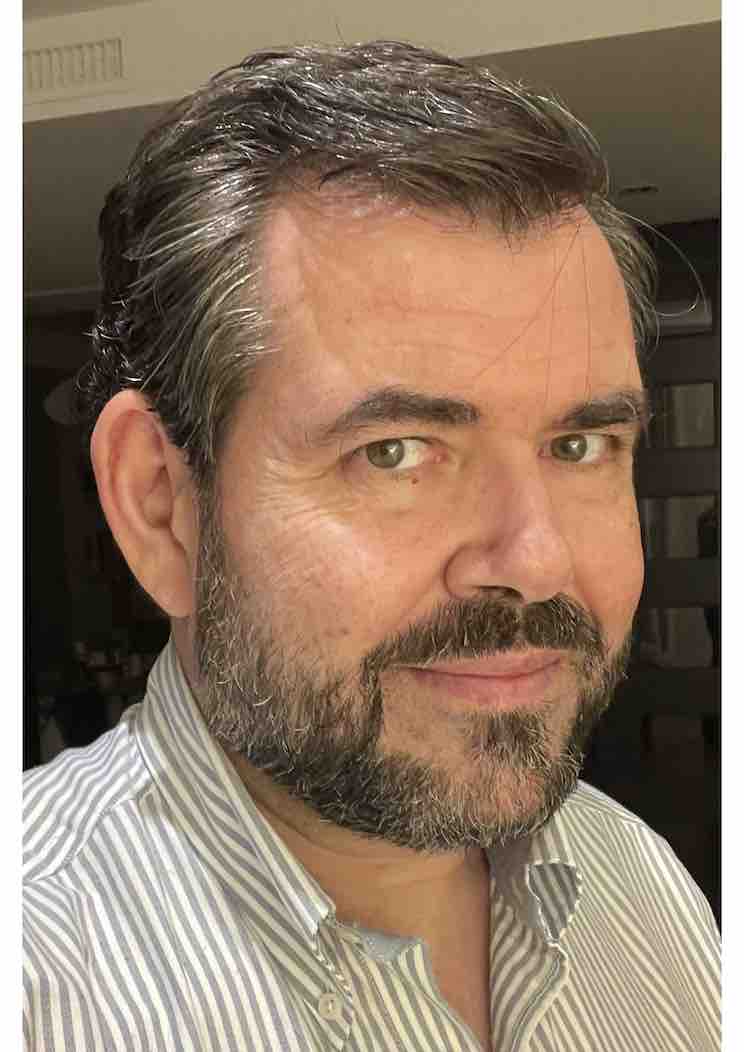}}]{Rafael Vazquez}{\space} (Senior Member, IEEE) received the electrical engineering and mathematics degrees from the University of Seville, Spain, and the M.Sc. and Ph.D. degrees in aerospace engineering from the University of California, San Diego.
He is currently Professor in the Aerospace Engineering Department of the University of Seville, Spain. His research interests include control theory, estimation, distributed parameter systems, and optimization, with applications to spacecraft and aircraft guidance, navigation and control and space surveillance and awareness. 
He currently serves as Associate Editor for Automatica and IEEE Control Systems Letters (L-CSS).
\end{IEEEbiography}

\begin{IEEEbiography}[{\includegraphics[width=1in,height=1.25in,clip,keepaspectratio]{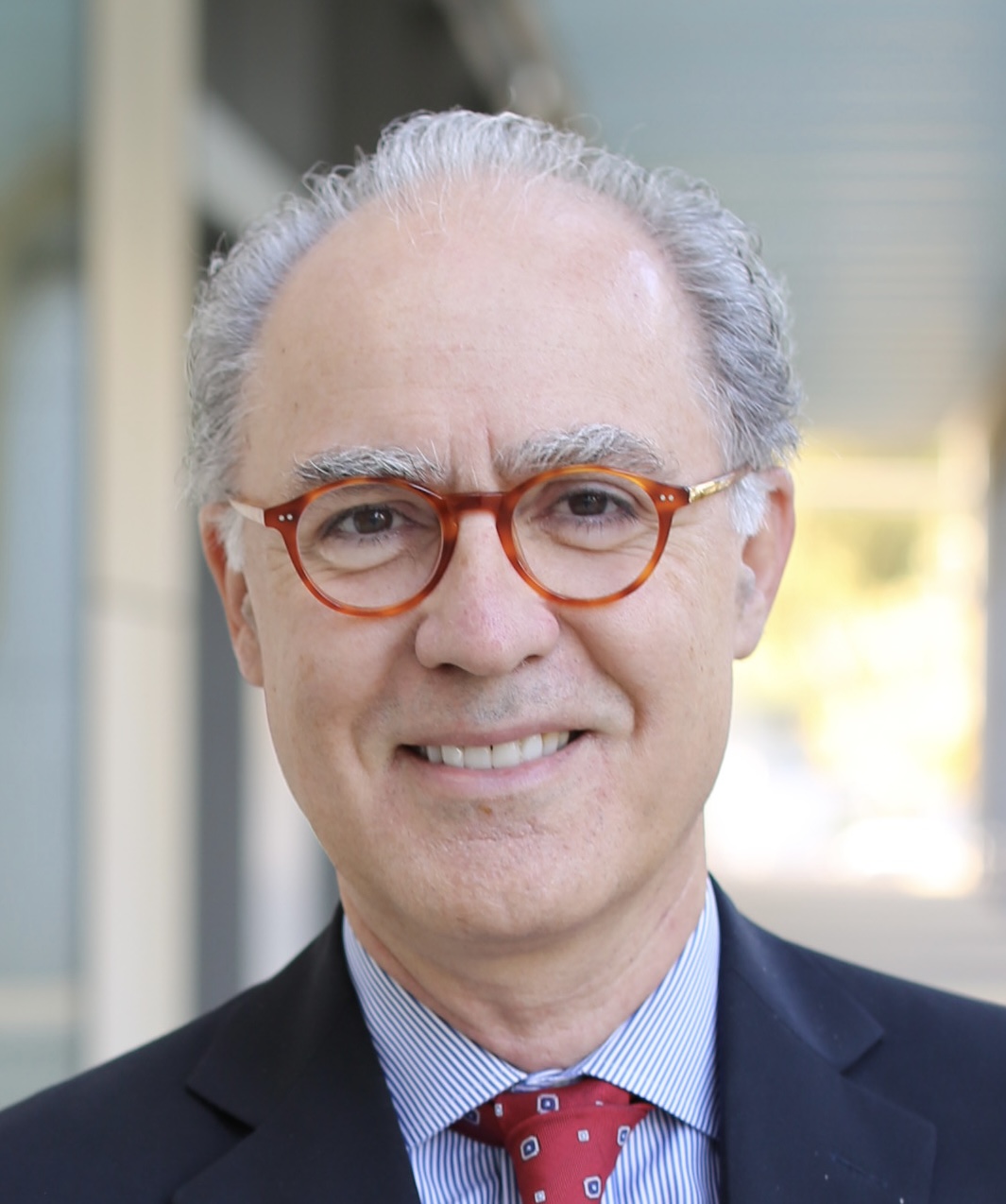}}]{Miroslav Krstic} is 
Sr.  Assoc. Vice Chancellor for Research at UC San Diego. He is Fellow of IEEE, IFAC, ASME, SIAM, AAAS, 
and 
Serbian Academy of Sciences and Arts. 
He has received the Bode Lecture Prize, Bellman Award,  Reid Prize,  Oldenburger Medal, Nyquist Lecture Prize, Paynter Award, Ragazzini  Award, IFAC Ruth Curtain Distributed Parameter Systems Award, IFAC Nonlinear Control Systems Award, IFAC Adaptive and Learning System Award, Chestnut textbook prize, AV Balakrishnan Award, CSS Distinguished Member Award, the PECASE, NSF Career, and ONR YIP,  Schuck (’96 and ’19) and Axelby paper prizes. 
\end{IEEEbiography}

\end{document}